\documentclass[11pt,a4paper]{llncs}

\usepackage[utf8x]{inputenc}
\usepackage[T1]{fontenc}

\usepackage{amsmath,amssymb,amsfonts}

\usepackage[numbers]{natbib}

\usepackage{tikz}
\usetikzlibrary{positioning,backgrounds,patterns,calc}
\tikzstyle{vertex}=[circle,draw=black,minimum size=8pt,inner sep=1pt]
\tikzstyle{vertex2}=[circle,draw=black,minimum size=15pt,inner sep=2pt]
\tikzstyle{edge}=[]
\tikzstyle{ypath}=[ultra thick]
\tikzstyle{dottedEdge}=[dotted,thick]
\tikzstyle{small-vertex}=[circle,draw=black,minimum size=6pt,inner sep=0pt,fill=white]
\tikzstyle{thinedges}=[draw=gray!30]
 \tikzstyle{boxes}=[draw,thick, rounded corners=3mm,text width=2.7cm,align=center,text opacity=1,fill opacity=1,fill=white]
\tikzstyle{unk}=[fill=gray!25!white]

\usepackage{mathtools}
\usepackage{xspace}

\spnewtheorem{observation}{Observation}{\bfseries}{\itshape}
\spnewtheorem{myclaim}{Claim}{\bfseries}{\itshape}

\newcommand{\plur}{\textnormal{\textsc{Plurality}}}

\newcommand{\cccav}{\textnormal{\textsc{C-Cons-Add}}\xspace}
\newcommand{\cdcav}{\textnormal{\textsc{C-Des-Add}}\xspace}
\newcommand{\cccdv}{\textnormal{\textsc{C-Cons-Del}}\xspace}
\newcommand{\cdcdv}{\textnormal{\textsc{C-Des-Del}}\xspace}

\newcommand{\ccccav}{\textnormal{Condorcet-\textsc{C-Cons-Add}}}

\newcommand{\domset}{\textnormal{\textsc{Dominating Set}}}
\newcommand{\indset}{\textnormal{\textsc{Independent Set}}}
\newcommand{\clique}{\textnormal{\textsc{Clique}}}

\newcommand{\fpt}{\textnormal{\textsf{FPT}}}
\newcommand{\Wone}{\textnormal{\textsf{W[1]}}}
\newcommand{\Wtwo}{\textnormal{\textsf{W[2]}}}

\newcommand{\cccauv}{\textnormal{\textsc{C-Cons-Add-Unlim}}}
\newcommand{\cdcauv}{\textnormal{\textsc{C-Des-Add-Unlim}}}
\newcommand{\cccduv}{\textnormal{\textsc{C-Cons-Del-Unlim}}}
\newcommand{\cdcduv}{\textnormal{\textsc{C-Des-Del-Unlim}}}

\newcommand{\mcccav}{\textnormal{\textsc{Min-C-Cons-Add}}}

\newcommand{\mcdcdv}{\textnormal{\textsc{Min-C-Des-Del}}}

\newcommand{\probTTSAT}{\textsc{(2-2)-3SAT}\xspace}

\usepackage{graphicx}
\usepackage{paralist}

\usepackage[math]{blindtext}

\usepackage{csquotes}

\usepackage{microtype}

\usepackage{url}
\makeatletter
\def\Url@twoslashes{\mathchar`\/\@ifnextchar/{\kern-.2em}{}}
\g@addto@macro\UrlSpecials{\do\/{\Url@twoslashes}}
\makeatother
\makeatletter
\g@addto@macro{\UrlBreaks}{\UrlOrds}
\makeatother

\usepackage{booktabs}
\usepackage{multirow}
\usepackage{tabulary}

\newcommand{\probDef}[3]{
  \begin{quote}
   #1\\
  \textbf{Input:} #2\\
  \textbf{Question:} #3
  \end{quote}
}

\usepackage{xcolor}

\usepackage[
bookmarks=false,
breaklinks=true,
colorlinks=true,
linkcolor=blue,
citecolor=green,
urlcolor=black,
pdfpagelayout=SinglePage,
pdfstartview=Fit,
pdfdisplaydoctitle,colorlinks,citecolor=green!60!black,linkcolor=red!60!black,pagebackref
]{hyperref}
\usepackage[all]{hypcap}

\usepackage{todonotes}
\newcommand{\inlinetodo}[1]{}

\usepackage[ruled,resetcount,algosection]{algorithm2e}
\makeatletter
\newcommand{\nosemic}{\renewcommand{\@endalgocfline}{\relax}}
\makeatother

\usepackage{cleveref}
\crefname{subsection}{Subsection}{Subsections}
\crefname{section}{Section}{Sections}

\crefname{table}{Table}{Tables}
\crefname{figure}{Figure}{Figures}
\crefname{algorithm}{Algorithm}{Algorithms}
\crefname{theorem}{Theorem}{Theorems}
\crefname{conjecture}{Conjecture}{Conjectures}
\crefname{definition}{Definition}{Definitions}
\crefname{corollary}{Corollary}{Corollary}
\crefname{proposition}{Proposition}{Propositions}
\crefname{observation}{Observation}{Observations}
\crefname{lemma}{Lemma}{Lemmas}
\crefname{example}{Example}{Examples}
\crefname{reduction}{Reduction}{Reductions}
\crefname{algorithm}{Algorithm}{Algorithms}
\crefname{appendix}{Appendix}{Appendices}
\crefname{claim}{Claim}{Claims}
\crefname{line}{line}{lines}
\crefname{myclaim}{Claim}{Claims}

\DeclareFontFamily{U}{MnSymbolC}{}
\DeclareSymbolFont{MnSyC}{U}{MnSymbolC}{m}{n}
\DeclareFontShape{U}{MnSymbolC}{m}{n}{
    <-6>  MnSymbolC5
   <6-7>  MnSymbolC6
   <7-8>  MnSymbolC7
   <8-9>  MnSymbolC8
   <9-10> MnSymbolC9
  <10-12> MnSymbolC10
  <12->   MnSymbolC12%
}{}
\DeclareMathSymbol{\powerset}{\mathord}{MnSyC}{180}

\usepackage{chngcntr}
\usepackage{apptools}

\newcommand{\N}{\mathbb{N}}

\newcommand{\NP}{\mathsf{NP}}
\newcommand{\PP}{\mathsf{P}}

\newcommand{\ON}{O}

\newcommand{\aq}{\Leftrightarrow}

\newcommand{\gap}{\textsf{gap}\xspace}

\DeclarePairedDelimiterX\Set[2]{\lbrace}{\rbrace}%
{\   #1 \  \,\delimsize| \ \,\mathopen{} #2 \ }
\DeclarePairedDelimiterX\Seta[2]{\langle}{\rangle}%
{\   #1 \  \,\delimsize| \ \,\mathopen{} #2 \ }
\DeclareMathOperator*{\argmax}{arg\,max}
\DeclareMathOperator*{\argmin}{arg\,min}

\usepackage[labelfont=bf]{caption}

\AtAppendix{\counterwithin{lemma}{section}}
\AtAppendix{\counterwithin{corollary}{section}}
\AtAppendix{\counterwithin{theorem}{section}}
\AtAppendix{\counterwithin{observation}{section}}
\AtAppendix{\counterwithin{example}{section}}
\AtAppendix{\counterwithin{proposition}{section}}



\usepackage{etoolbox}

\newcommand{\myparagraph}[1]{\smallskip \noindent \textbf{#1}}

\newcommand{\appendixproof}[2]{%
  \gappto{\appendixProofText}{\subsection{Proof of \cref{#1}}\label{proof:#1}#2}
}
\newcommand{\appendixproofwtext}[3]{%
  \gappto{\appendixProofText}{\subsection{Proof of \cref{#1}}\label{proof:#1}
  \noindent \textbf{\cref{#1}.}
  #2
  #3
  }
}

\newcommand{\appendixproofwithoutstatement}[2]{
  \gappto{\appendixProofText}{
    \label{proof:#1}
    #2
  }
}

\title{On the Computational Complexity of Variants of Combinatorial Voter Control in Elections}

\author{
Leon Kellerhals
\and
Viatcheslav Korenwein
\and
Philipp Zschoche\thanks{PZ was supported by the Stiftung Begabtenf\"orderung berufliche Bildung (SBB).}
\and
Robert~Bredereck\thanks{RB was from September 2016 to September 2017 on postdoctoral leave at the University of Oxford (GB), supported by the DFG fellowship BR 5207/2.}
\and
Jiehua Chen
}
\institute{
	TU Berlin, Germany \\
	\texttt{\{leon.kellerhals, viatcheslav.korenwein, zschoche\}@campus.tu-berlin.de}\\
        \texttt{\{robert.bredereck, jiehua.chen\}@tu-berlin.de}
}

\begin{document}

\maketitle

\begin{abstract}
	Voter control problems model situations in which an external agent tries to 
	affect the result of an election by adding or deleting the fewest number of voters. 
	The goal of the agent is to make a specific candidate either win (\emph{constructive} control) or lose (\emph{destructive} control) the election. 
        We study the constructive and destructive voter control problems when
        adding and deleting voters have a \emph{combinatorial flavor}:
        If we add (resp.\ delete) a voter~$v$,
        we also add (resp.\ delete) a bundle~$\kappa(v) $ of voters that are associated with~$v$.
        While the bundle~$\kappa(v)$ may have more than one voter,
        a voter may also be associated with more than one voter.
	We analyze the computational complexity of the four voter control problems for the Plurality rule.

        We obtain that, in general, making a candidate lose is computationally easier than making her win.
        In particular, if the bundling relation is symmetric (i.e.\ $\forall w\colon w \in \kappa(v) \Leftrightarrow v \in \kappa(w) $),
	and if each voter has at most two voters associated with him,
 	then destructive control is polynomial-time solvable
	while the constructive variant remains $\NP$-hard.
        Even if the bundles are disjoint
	(i.e.\ $\forall w\colon w \in \kappa(v) \Leftrightarrow \kappa(v) = \kappa(w) $),
        the constructive problem variants remain intractable. 
	Finally, the minimization variant of constructive control by adding voters does not admit 
        an efficient approximation algorithm, unless $\PP = \NP$.
\end{abstract}

\section{Introduction}

Since the seminal paper by \citet{bartholdi1992hard} on controlling 
an election by adding or deleting the fewest number of voters or candidates 
with the goal of making a specific candidate to win (\emph{constructive control}),
a lot of research has been devoted to the study of control for different voting rules~\cite{FalHemHemRot2009,ErdFelRotSch2015,fen-liu-lua-zhu:j:parameterized-control,liu-zhu:j:maximin,bet-uhl:j:parameterized-complexity-candidate-control,HemFavMent2016},
on different control modes~\cite{FalHemHem2011,FalHemHem2015},
or even on other controlling goals~(e.g. aiming at several candidates' victory or a specific candidate's defeat) \cite{hemaspaandra2007anyone,ProRosZoh2007}.
Recently, \citet{bulteau2015combinatorial} introduced combinatorial structures to constructive control by adding voters: When a voter is added, a bundle of other voters is added as well.
A combinatorial structure of the voter set allows us to model situations
where an external agent hires speakers to convince whole groups of people
to participate in (or abstain from) an election.
In such a scenario, convincing a whole group of voters comes at the fixed cost of paying a speaker.
\citet{bulteau2015combinatorial} model this by defining a bundle of associated voters for each voter which will be convinced to vote ``for free'' when this voter is added or deleted.
Moreover, the bundles of different voters could overlap.
For instance, convincing two bundles of two voters each to participate in the election could result in adding a total of four, three or even just two voters.

We extend the work of \citet{bulteau2015combinatorial} 
and investigate the cases where the agent wants to make a specific candidate win or lose 
by adding (resp.\ deleting) the fewest number of bundles.
We study one of the simplest voting rules, the Plurality rule, 
where each voter gives one point to his favorite candidate, 
and the candidate with most points becomes a winner.
Accordingly, an election consists of a set~$C$ of candidates and a set~$V$ of voters
who each have a favorite candidate.
Since real world elections typically contain only a small number of candidates, 
and a bundle of voters may correspond to a family with just a few members,
we are especially interested in situations where the election has only few candidates
and the bundle of each voter is small.
Our goal is to ensure that a specific candidate $ p $ becomes a winner 
(or a loser) of a given election, by convincing as few voters from 
an unregistered voter set~$W$ as possible (or as few 
voters from $ V $ as possible), together with the voters in their bundles,
to participate (or not to participate) in the election. 
We study the combinatorial voter control problems from both the classical and the parameterized complexity point of view.
We confirm \citeauthor{bulteau2015combinatorial}'s conjecture~\cite{bulteau2015combinatorial} that 
%
for the Plurality rule, 
the three problem variants: combinatorial constructive control by deleting voters and
combinatorial destructive control by adding (resp.\ deleting) voters,
behave similarly in complexity to the results of 
combinatorial constructive control by adding voters:
They are $\NP$-hard and intractable even for very restricted cases.
We can also identify several special cases, where the complexity of the four problems behave differently.
For instance, we find that constructive control tends to be computationally harder than destructive control.
We summarize our results in \cref{tab:results}.

\myparagraph{Related Work.}
\citet{bartholdi1992hard} introduced the complexity study of election control problems 
and showed that for the Plurality rule, 
the non-combinatorial variant of the voter control problems
can be solved in linear time by a using simple greedy strategy.
We refer the readers to the work by \citet{RotSch2013,FalRot2016} for general expositions on election control problems.

In the original election control setting, 
a unit modification of the election concerns usually a single voter or candidate.
The idea of adding combinatorial structure to election voter control was initiated by \citet{bulteau2015combinatorial}:
Instead of adding a voter at each time, one adds a ``bundle'' of voters to the election, 
and the bundles added to the election could intersect with each other. 
They showed that combinatorial constructive control by adding the fewest number of bundles becomes notorious hard, even for the Plurality rule and for only two candidates.
\citet{Chen2015} mentioned that even if each bundle has only two voters and the underlying bundling graph is acyclic, the problem still remains $\NP$-hard.
\citet{bulteau2015combinatorial} and \citet{Chen2015} conjectured that 
\begin{quote}
  ``the combinatorial addition of voters for destructive control, and combinatorial deletion of voters for either constructive or destructive control behave similarly to combinatorial addition of voters
for constructive control.''
\end{quote}

The combinatorial structure notion for voter control has also been extended to candidate control~\cite{CheFalNieTal2015} and electoral shift bribery~\cite{BreFalNieTal2016a}.

\myparagraph{Paper Outline.}
In \cref{sec:preliminaries}, we introduce the notation used throughout the paper. 
In \cref{sec:central_problem} we formally define the four problem variants, 
summarize our contributions, present results in which the four problem variants
(constructive or destructive, adding voters or deleting voters) behave similarly,
and provide reductions between the problem variants. 
\cref{sec:des_vs_cons,sec:disjoint,sec:inapproximability} present our main results on
three special cases (1) when
the bundles and the number of candidates are small, 
(2) when the bundles are disjoint,
and (3) when the solution size could be unlimited.
We conclude in \cref{sec:conclusion} with several future research directions.

\section{Preliminaries}
\label{sec:preliminaries}
The notation we use in this paper is based on \citet{bulteau2015combinatorial}.
We assume familiarity with standard notions regarding algorithms and complexity theory.
For each $ z \in \N $ we denote by $ [z] $ the set~$ \{ 1, \dots, z \} $.

\myparagraph{Elections.}
An \emph{election} $ E = (C, V) $ consists of a set $ C $ of $ m $ \emph{candidates} 
and a set $ V $ of \emph{voters}. 
Each voter $ v \in V $ has a favorite candidate~$c$ and we call voter~$v$ a $c$-voter.
Note that since we focus on the Plurality rule, 
we simplify the notion of the preferences of voters in an election to the favorite candidate of each voter.
For each candidate~$c\in C$ and each subset $V'\subseteq V$ of voters, 
her \emph{(Plurality) score~$s_c(V')$} is defined as the number of voters from $V'$ that have her as favorite candidate.
We say that a candidate~$c$ is a \emph{winner} of election~$(C,V)$ 
if $c$ has the highest score~$s_c(V)$.
For the sake of convenience, for each $ C' \subseteq C $, a $ C' $-voter 
denotes a voter whose favorite candidate belongs to $C' $.



\myparagraph{Combinatorial Bundling Functions.}
Given a voter set $ X $, a \emph{combinatorial bundling function} $ \kappa\colon X \to 2^X $ 
(abbreviated as \emph{bundling function}) is a function that assigns a set of voters
to each voter; we require that $x\in \kappa(x)$. 
For the sake of convenience, for each subset $ X' \subseteq X $, we define 
$ \kappa(X') = \bigcup_{x \in X'} \kappa(x) $. 
For a voter $ x \in X $, $ \kappa(x) $ is named \emph{$ x $'s bundle}; 
$ x $ is called the \emph{leader} of the bundle.
We let $ b $ denote the \emph{maximum bundle size} of a given $ \kappa $. 
Formally, $ b = \max_{x \in X} | \kappa(x) | $.
One can think of the bundling function as subsets of voters that can be added at a unit cost
(e.g. $ \kappa(x) $ is a group of voters influenced by $ x $). 

\myparagraph{Bundling graphs.}
The \emph{bundling graph} of an election is a model of how the voter's bundle functions
interact with each other. 

Let $ \kappa\colon X \to 2^X $ be a bundling function. 
The \emph{bundling graph} $ G_\kappa = (V(G_\kappa), E(G_\kappa)) $ 
is a simple, directed graph:
For each voter $ x \in X $ there is a vertex $ u_x \in V(G_\kappa) $.
For each two distinct voters $  y, z \in X $
with $ z \in \kappa(y) $, there is an arc $ (u_y, u_z) \in E(G_\kappa) $.

We consider three special cases of bundling functions/graphs which we think are natural in real world.
We say that a bundling function~$\kappa$ is \emph{symmetric} if for each two distinct voters $ x, y \in X $, 
it holds that $ y \in \kappa(x) $ if and only if $ x \in \kappa(y) $. 
The bundling graph for a symmetric bundling function always has two directed arcs connecting each two vertices.
Thus, we can assume the graph to be undirected.

We say that $ \kappa $ is \emph{disjoint} if for each two distinct voters~$x,y\in X$,
it holds that either $\kappa(x)=\kappa(y)$ or $\kappa(x)\cap \kappa(y)=\emptyset$.
It is an easy exercise to verify that disjoint bundling functions are symmetric
and the corresponding undirected bundling graphs consist only of disjoint complete subgraphs.

We say that $ \kappa $ is \emph{anonymous} if for each two distinct voters $ x $ and $ y $ with 
the same favorite candidate, 
it holds that $ \kappa(x) = \kappa(y) $, 
and that for all other voters~$z$ we have $ x \in \kappa(z) $ if and only if $  y \in \kappa(z) $.

\myparagraph{Parameterized Complexity.}
An instance $ (I, r) $ of a \emph{parameterized problem} consists of the actual 
instance $ I $ and of an integer $ r $ referred to as the \emph{parameter} 
\cite{downey2013fundamentals,flum2006parameterized,niedermeier2006invitation}.
A parameterized problem is called \emph{fixed-parameter tractable} (in \fpt{})
if there is an algorithm that solves each instance $ (I, r) $ in $ f(r) \cdot |I|^{\ON(1)} $ time,
where $ f $ is a computable function depending only on the parameter $ r $.

There is also a hierarchy of hardness classes for parameterized problems,
of which the most important ones are \Wone{} and \Wtwo{}. 
One can show that a parameterized problem $ L $ is (presumably) not in \fpt{} by
devising a \emph{parameterized reduction} from a \Wone{}-hard or a \Wtwo{}-hard 
problem to $ L $.
A parameterized reduction from a parameterized problem $ L $ to another 
parameterized problem $ L' $ is a function that acts as follows:
For two computable functions $ f $ and $ g $, given an instance $ (I, r) $ of problem~$L$, 
it computes in $ f(r) \cdot |I|^{\ON(1)} $ time an instance $ (I', r') $ of problem $ L' $ so that
$ r' \le g(r) $ and that $ (I, r) \in L $ if and only if $ (I', r') \in L' $. 
For a survey of research on parameterized complexity in computational social choice, we refer to \citet{betzler2012studies} and \citet{bredereck2014parameterized}.

\section{Central Problem}
\label{sec:central_problem}
We consider the problem of \emph{combinatorial voter control} in four variants. 
The variants differ in whether they are \emph{constructive} or \emph{destructive}, meaning
that the goal is to make one selected candidate win or lose the election. 
This goal can be achieved by either \emph{adding} voters to or \emph{deleting} voters from the given 
election. 
Due to space constraints, we only provide the definition of constructive control. Destructive control is defined analogously.

\medskip
\noindent \textbf{\textsc{Combinatorial Constructive Control by Adding}}\\
(resp.~\textbf{\textsc{Deleting}})~\textbf{\textsc{Voters}} [\textbf{\textsc{C-Cons-Add}} (resp.\ \textbf{\textsc{C-Cons-Del}})]

\noindent  \begin{tabular}{@{}lp{.87\linewidth}@{}}
  	\textbf{Input:}& An election $ E = (C, V) $, a set $ W $ of unregistered voters with 
	$ V \cap W = \emptyset $, a bundling function $ \kappa\colon W \to 2^W $ (resp.\ $\kappa\colon V\to 2^V$), a preferred winner 
	$ p \in C $, and an integer~$ k \in \N $. \\
  	\textbf{Quest.:}& Is there a size-at-most~$k$ subset $ W' \subseteq W $ (resp.\ $V'\subseteq V$) of voters such that $ p $ wins the election $(C, V \cup \kappa(W')) $ (resp.\ $(C, V \setminus \kappa(V'))$)?
\end{tabular}
\medskip

Throughout this work, when speaking of the \emph{``adding''} or \emph{``deleting''} variants, 
we mean those variants in which voters are added or, respectively, deleted.
In similar fashion, we speak of the \emph{constructive} and \emph{destructive} (abbr.\ by ``\textsc{Cons}'' and  by ``\textsc{Des}'', respectively) problem variants. 
Further, we refer to the set $ W' $ of voters as the solution for the ``adding'' variants 
(the set $ V' $ of voters for the ``deleting'' variants, respectively) and denote $ k $
as the solution size.

\myparagraph{Our Contributions.}
We study both the classical and the parameterized complexity of the four voter control variants.
We are particularly interested in the real-world setting where the given election has a small number of candidates and where only a few voters are associated to a voter.
On the one hand, we were able to confirm the conjecture given by \citet{bulteau2015combinatorial} and \citet{Chen2015}
that when parameterized by the solution size, 
\cccdv, \cdcav, and \cdcdv{} are all intractable even for just two candidates
or for bundle sizes of at most three, 
and that when parameterized by the number of candidates, 
they are fixed-parameter tractable for anonymous bundling functions.
On the other hand, we identify interesting special cases where the four problems differ in their computational complexity. 
We conclude that in general, destructive control tends to be easier than constructive control:
For symmetric bundles with at most three voters,
\cccav{} is known to be $\NP$-hard, 
while both destructive problem variants are polynomial-time solvable.
For disjoint bundles, constructive control is parameterized intractable (for the parameter ``solution size~$k$''), while destructive control is polynomial-time solvable.
Unlike for \cccdv{}, a polynomial-time approximation algorithm for \cccav{} does not exist, unless $ \PP = \NP$.
Our results are gathered in \cref{tab:results}.

\begin{table}[t]	
  \centering
  \caption{Computational complexity of the four combinatorial voter control variants with the Plurality rule.
	The parameters are 
	``the solution size $ k $'', 
	``the number $ m $ of candidates'' and
	``the maximum bundle size $ b $''.
	We refer to $|I|$ as the instance size.
	The rows distinguish between different maximum bundle sizes
	$ b $ and the number~$m$ of candidates. 
        All parameterized intractability results are for the parameter ``solution size~$k$''.
        \textsf{ILP-FPT} means \fpt{} based on a formulation as an integer linear program and the result is for the parameter~``number~$m$ of candidates''. 
      }\label{tab:results}
  \resizebox{\textwidth}{!}{
  \begin{tabular}{@{}l@{\hspace{1ex}}l@{\hspace{1ex}}l@{\hspace{2ex}}ll@{\hspace{1ex}}l@{\hspace{2ex}}l@{}}
    \toprule
    & \cccav & \cccdv && \cdcav & \cdcdv & References\\\midrule
    \underline{$\kappa$ symmetric} \\
    \; $b=2$& $O(|I|)$ & $\PP$ && $O(m|I|)$ & $O(m|I|)$ & Obs~\ref{thm:b2symm}, Thm~\ref{thm:b2symm-del}\\
    &&&&&&Thm~\ref{thm:symmdes}\\
    \; $b=3$ &\\
    \;\;\; $m=2$ & $O(|I|^5)$ & $O(|I|^5)$ && $O(|I|^5)$  & $O(|I|^5)$ & Thm~\ref{thm:m2symm},\\
    &&&&&& Cor~\ref{thm:cons_des_poly}+\ref{thm:des_poly}\\
    \;\;\; $m$ unbounded & $\NP$-h & $\NP$-h && $O(m|I|^5)$ & $O(m|I|^5)$ & Obs~\ref{thm:cons_hard}, Prop~\ref{prop:des_hard},\\
    &&&&&& Cor~\ref{thm:des_poly}\\
     \; $b$ unbounded & \\
    \;\;\; $m=2$ & \Wtwo-h & \Wtwo-h && \Wtwo-h & \Wtwo-h & \cite{bulteau2015combinatorial}, Thm~\ref{thm:similar_summary}\\
    \;\;\; $m$ unbounded and \\
    \;\;\;\;\; $\kappa$ disjoint& \Wone-h & \Wtwo-h && $O(m|I|)$ & $O(m|I|)$ & Thm~\ref{thm:disjoint_hard}+\ref{thm:symmdes}\\
    &&&&&\\[-1ex]
    \underline{$\kappa$ anonymous} & \textsf{ILP-FPT} & \textsf{ILP-FPT} && \textsf{ILP-FPT} & \textsf{ILP-FPT} & Thm~\ref{thm:similar_summary}\\
    &&&&&\\[-1ex]
    \underline{$\kappa$ arbitrary}\\
    \; $b=3$, $m=2$ & \Wone-h & \Wone-h& & \Wone-h& \Wone-h & Thm~\ref{thm:similar_summary}\\\bottomrule
  \end{tabular}
  }

\end{table}

The following theorem summarizes the conjecture given by \citet{bulteau2015combinatorial} and \citet{Chen2015}. The proofs are deferred to \cref{app:w2,app:w1,app:ilp}.
\begin{theorem}
	\label{thm:similar_summary}
	All four combinatorial voter control variants are 
	\begin{compactenum}[(i)]
		\item \Wtwo{}-hard with respect to the solution size~$k$, 
                even for only two candidates and for symmetric bundling functions~$\kappa$ 
		\item \Wone{}-hard with respect to the solution size~$k$, even 
                for only two candidates and for bundle sizes of at most three.
		\item fixed-parameter tractable with respect to the number~$m$ of candidates 
                if the bundling function $ \kappa $ is anonymous.
	\end{compactenum}
\end{theorem}

\myparagraph{Relations between the four problem variants.}
We provide some reductions between the problem variants. 
They are used in several sections of this paper. 
The key idea for the reduction from destructive control to constructive control
is to guess the candidate that will defeat the distinguished candidate 
and ask whether one can make this candidate win the election.
The key idea for the reduction from the ``deleting'' to the ``adding'' problem variants
is to build the ``complement'' of the registered voter set. 
\newcommand{\thmturAtext}{
For each $X\in $ \textsc{\{Add, Del\}}, \textsc{C-Des-$X$} with $m$ candidates 
is Turing reducible to \textsc{C-Cons-$X$} with two candidates.
}
\newcommand{\thmturBtext}{
For each $Y\in $ \textsc{\{Cons, Des\}}, \textsc{C-$Y$-Del} with two candidates 
is many-one reducible to \textsc{C-$Y$-Add} with two candidates.
}
\begin{proposition}\label{thm:tur}
		\thmturAtext
		\thmturBtext
	All these reductions
	preserve the property of symmetry of the bundling functions.
\end{proposition}
\appendixproof{thm:tur}{
\begin{proposition}
 \thmturAtext
\end{proposition}
\begin{proof}
	First, we provide a Turing reduction from \cdcav{} to \cccav{}.
	Let $ I = (E = (C, V), W, \kappa, p, k)$ be a \cdcav{} instance with
	candidate set~$C = \{p, $ $ g_1, \dots, g_{m-1} \}$.
	We compose $m-1$ instances of the problem \cccav{} 
	$ J_i = (E_i = (C_i, V_i), W', \kappa_i, g_i, k)$,
	where 
	\begin{itemize}
	\item $ C_i \coloneqq \{ p, g_i \}$, 
	\item $ V_i \coloneqq \{v \in V\mid v\text{'s more preferred candidate is either } p \text{ or } g_i\}
		\cup \{ v_d \} $, where $ v_d $ is an additional $ p $-voter,
	\item $ W_i \coloneqq \{w \in W\mid w\text{'s more preferred candidate is either } p \text{ or } g_i\}$, and
	\item $ \kappa_i \colon W_i \to 2^{W_i}\text{ with } \kappa_i(w) \coloneqq \kappa(w) \cap W_i $.
	\end{itemize}
	We show that $ I $ has a solution of size at most $ k $ if and only if at least
	one of the instances $ J_i $ has a solution of size at most $ k $.

	For the ``only if'' part, let $ W'_i \subseteq W_i $ be a solution for $ J_i $. 
	Since $J_i$ has only two candidates~$p$ and $g_j$, this implies $ s_p(V_i \cup \kappa_i(W'_i)) \leq s_{g_i}(V_i \cup \kappa_i(W'_i)) $
	in $ J_i $. 
	Since we added an additional $ p $-voter to $ J_i $, 
	$ s_p(V \cup \kappa(W'_i)) < s_{g_i}(V \cup \kappa(W'_i))$ in $ I $.
	Thus $ p$ loses election~$(C, V \cup \kappa(W'_i)) $ and $ W'_i \subseteq W $
	is a solution for $ I $.	

	For the ``if'' part, let $ W' \subseteq W $ be a solution for $ I $, meaning that $p$ loses election~$(V \cup \kappa(W'))$.
	Thus, there exists a $ g_i \in C $ with 
	$ s_{g_i}(V \cup \kappa(W')) > s_p(V \cup \kappa(W')) $.
	Let $ W'_i \coloneqq W' \cap W_i $.
	Since $ g_i $ and $ p $ are the only candidates in $ C_i $, 
	and since $ V_i $ has one additional $ p $-voter compared to $ V $,
	it follows that $ s_{g_i}(V_i \cup \kappa_i(W'_i)) \ge s_p(V_i \cup \kappa_i(W'_i)) $.
	Thus, $ g_i$ wins election~$(C_i, V_i \cup \kappa_i(W'_i)) $, 
	and $ W'_i $ is a solution for $ J_i $.

	For the Turing reduction from \cdcdv{} to \cccdv{}, the construction of the 
	reduction is similar to the one from \cdcav{} to \cccav{}.
	Set the bundling functions 
	$ \kappa_i \colon V_i \to 2^{V_i} $ with $\kappa_i(v) \coloneqq \kappa(v) \cap V_i $
        and $\kappa_i (v_d) = \{v_d\}$.
\qed\end{proof}

\begin{proposition}	
	\thmturBtext
\end{proposition}
\begin{proof}
	First, we provide a polynomial-time reduction from \cccdv{} to \cccav{}. 
	Let $I = (E = (C, V), \kappa, p, k)$ be a \cccdv{} instance with
	$C=\{p,g\}$.
	We define $ \overline V $ to be the \emph{complement voter set}, 
	that is, it contains the same 
	voters as $ V $ but all $ p $-voters become $ g $-voters and 
	all $ g $-voters become $ p $-voters. 
	We construct a \cccav{} instance 
	$ J = (E = (C, V), W, \kappa, p, k) $,
	where $ W \coloneqq \overline V $.
	Clearly, the construction of $ J $ can be implemented in polynomial time.
	We now claim that $ V' $ is a size-$k$ solution for $ I $ if and only if $ W' \coloneqq \overline{V'} $
	is a size-$k$ solution for $ J $.

	First, we observe that for $ V' $ to be a solution for $ I $, it must hold that 
	$ s_p(V) - s_p(\kappa(V')) \ge s_g(V) - s_g(\kappa(V')) $. 
	Similarly for $ W' $ to be a solution for $ I $, it must hold that 
	$ s_p(V) + s_p(\kappa(W')) \ge s_g(V) + s_g(\kappa(W')) $. 
	As per definition of the complement voter set and because of $ \overline{V'} = W' $, 
	we know that 
		\begin{align*}
                  s_p(\kappa(V')) = s_g(\kappa(W')) \text{ and } s_g(\kappa(V')) = s_{p} (\kappa(W')).
                \end{align*}
	Thus, it holds that 
	\begin{alignat*}{2}
          & s_p(V) - s_p(\kappa(V')) &~\ge~ & s_g(V) - s_g(\kappa(V')) \\
          \text{ if and only if }~&          s_p(V) + s_g(\kappa(V')) &~\ge~& s_g(V) + s_p(\kappa(V'))\\
          \text{ if and only if }~&       s_p(V) +  s_p(\kappa(W')) &~\ge~& s_g(V) + s_g(\kappa(W')),
        \end{alignat*} 
	implying that $ V' $ is a solution for $ I $ if and only if $ W' $ is a solution
	for $ J $.

	The reduction from \cdcdv{} to \cdcav{} works analogously. 
\qed\end{proof}
}



\section{Controlling Voters with Symmetric and Small Bundles}
\label{sec:des_vs_cons}
In this section, we study combinatorial voter control when the voter bundles are symmetric and small.
This could be the case when a voter's bundle models his close friends (including himself), close relatives, or office mates.
Typically, this kind of relations is symmetric, and the number of friends, relatives, or office mates is small. 
We show that for symmetric bundles and for bundles size at most three, 
both destructive problem variants become polynomial-time solvable, 
while both constructive variants remain $\NP$-hard.
However, if there are only two candidates,
then we can use dynamic programming to also solve the constructive control variants in polynomial time.
If we restrict the bundle size to be at most two, 
then all four problem variants can be solved in polynomial time via simple greedy algorithms.


As already observed in \cref{sec:preliminaries}, 
we only need to consider the undirected version of the bundling graph for symmetric bundles.
Moreover, if the bundle size is at most two, 
then the resulting bundling graph consists of only cycles and trees.
However, \citet{bulteau2015combinatorial}
already observed that \cccav{} is $\NP$-hard even if the resulting bundling graph solely consists of cycles, 
and \citet{Chen2015} observed that \cccav{} remains $\NP$-hard even if the resulting bundling graph consists of only directed trees of depth at most three.

\newcommand{\thmconshardtext}
{\cccav{} is $\NP$-hard even for symmetric bundling functions with maximum bundle size $b=3$.}
\begin{observation}
	\label{thm:cons_hard}
	\thmconshardtext
\end{observation}

\appendixproofwtext{thm:cons_hard}{\thmconshardtext}{

\begin{proof}
	\citet[Theorem 7]{bulteau2015combinatorial} have shown $\NP$-hardness of \cccav{} for full-$d$ bundling 
	functions\footnote{\label{ft:fulld}Full-$d$ bundling functions are defined by \citeauthor{bulteau2015combinatorial}~\citep[Section 2]{bulteau2015combinatorial}.} 
	and maximum bundle size $b=3$.
	Since a full-$d$ bundling function is symmetric~\cite[Observation 1]{bulteau2015combinatorial}. 
	Thus, $\NP$-hardness of \cccav{} for symmetric bundling functions and $b \leq 3$ follows.
\qed\end{proof}
}

It turns out that the reduction used by \citet{bulteau2015combinatorial} to show 
\cref{thm:cons_hard} can be adapted to show $\NP$-hardness for the deleting case. 

\newcommand{\propdeshardtext}
{\cccdv{} is $\NP$-hard even for symmetric bundling functions with maximum bundle size $b=3$.}

\begin{proposition}
  \label{prop:des_hard}
  \propdeshardtext
\end{proposition}
\appendixproofwtext{prop:des_hard}{\propdeshardtext}{
\begin{proof}
	\citet[Theorem 7]{bulteau2015combinatorial} have shown $\NP$-hardness of \cccav{} for full-$d$ bundling
	functions\footnote{See \cref{ft:fulld}.}
	and maximum bundle size $b \leq 3$.
        Their idea was to construct a cycle for each variable~$x$: 
        The cycle contains vertices that correspond to the clauses containing either $x$ or $\overline{x}$ and are connected through some $p$-voters in such a way 
        that one must take the vertices corresponding to the clauses with the same literal, that is, either $x$ or $\overline{x}$.
        
        We utilize this construction to show the hardness for the deleting case by reducing from the following $\NP$-complete \textsc{3SAT} problem~\cite{bulteau2015combinatorial}.

        \probDef
        {\probTTSAT}
        {A collection~$\mathcal{F}$ of size-two-or-three clauses over the variable set~$\mathcal{X} = \{x_1, \ldots, x_n\}$, such that each clause has either two or three literals, and each variable appears exactly four times, twice as a positive literal and twice as a negative literal.}
        {Is there a truth assignment that satisfies all the clauses in $\mathcal{F}$?}

       Let $I=(\mathcal{F}, \mathcal{X})$ be a \probTTSAT instance.
       Now, we construct an instance for \cccdv as follows:
       Let $p$ be the candidate whose victory we want to ensure 
       and let $d$ be the winner of the original election.
       For each clause~$f_j=(\ell^{1}_j\vee \ell^{2}_j \vee \ell^{3}_j)\in \mathcal{F}$, 
       we introduce a candidate~$c_j$.
       Thus, the candidate set is~$C=\{p,d\}\cup \{c_j \mid f_j \in \mathcal{F}\}$.
       
       The voter set is composed of three groups:
       \begin{enumerate}
         \item For each variable~$x_i$,
       for each clause~$f_j$ that contains $x_i$ as a literal (that is, either positive or negated variable), 
       we introduce a $d$-voter, denoted as $v^{j}_i$ (we call him a \emph{variable voter}), 
       and a $c_j$-voter, denoted as $u^{j}_i$ (we call him a \emph{clause voter}).
       Now, let $f_j$, $f_r$, $j< r$, 
       be the clauses containing~$x_i$, 
       and let $f_s$, $f_t$, $s<t$ be the two clauses containing $\overline{x}_i$. 
       We construct the bundles of these the clause voters corresponding to $f_j, f_r, f_s, f_t$ and the variable voters 
       such that the bundling graph forms a cycle with the $d$-voters between each two clause voters.
       More precisely, 
       let 
       \begin{alignat*}{4}
         &\kappa(u^j_i) &&\coloneqq \{u^{j}_i, v^j_i, v^{s}_i\}, & \quad \kappa(v^{s}_i) &&\coloneqq \{v^{s}_i, u^{j}_i, u_i^{s}\}\\
         &\kappa(u^s_i) &&\coloneqq \{u^{s}_i, v^s_i, v^{r}_i\}, & \quad \kappa(v^{r}_i) &&\coloneqq \{v^{r}_i, u^{s}_i, u_i^{r}\}\\
         &\kappa(u^r_i) &&\coloneqq \{u^{r}_i, v^r_i, v^{t}_i\}, & \quad \kappa(v^{t}_i) &&\coloneqq \{v^{t}_i, u^{r}_i, u_i^{t}\}\\
         &\kappa(u^t_i) &&\coloneqq \{u^{t}_i, v^t_i, v^{j}_i\}, & \quad \kappa(v^{j}_i) &&\coloneqq \{v^{j}_i, u^{s}_i, u_i^{j}\}.
       \end{alignat*}
       
       \item For each clause~$C_j$ that contains only two literals, we introduce a $c_j$-voter~$u^j_c$ and set his bundle to be the
             singleton~$\kappa(u^j_c) = \{u^j_c\}$. 

       \item We introduce two $p$-voters~$w^{+}_p, w^{-}_p$ and two $d$-voters $v^{+}_d, v^{-}_d$ with the following bundles:
      \begin{alignat*}{4}
         &\kappa(w^{+}_p) &&\coloneqq \{w^{1}_p, v^{+}_d\}, & \quad \kappa(w^{-}_p) &&\coloneqq \{w^{-}_p, v^{-}_d\},\\
         & \kappa(v^{+}_d) &&\coloneqq \{v^{+}_d, w^{+}_p\},  &\quad \kappa(v^{-}_d) && \coloneqq \{v^{-}_d,  w^{-}_p\}.
       \end{alignat*}
      
       \end{enumerate}
       The bundling graph for the first group of voters is depicted in the left figure and 
       the bundling graph for the second group of voters is depicted in the middle figure, and
       the bundling graph for the last group of voters is depicted in the right figure.

       {\centering
       \begin{tikzpicture}[every node/.style={draw=black,thick,circle,inner sep=0pt}]
        \def \n {8}
        \def \radius {2cm}
        \def \margin {13}
        \def \nodesize {0.85cm}
        
        \def \s {1} \node[draw, circle,fill=black!30][minimum size=\nodesize] at ({360/\n * (\s - 1)}:\radius) {$u_i^{j}$};
        \draw[-, >=latex] ({360/\n * (\s - 1)+\margin}:\radius) 
        arc ({360/\n * (\s - 1)+\margin}:{360/\n * (\s)-\margin}:\radius);

        \def \s {2}
        \node[draw, circle][minimum size=\nodesize] at ({360/\n * (\s - 1)}:\radius) {$v^{s}_i$};
        \draw[-, >=latex] ({360/\n * (\s - 1)+\margin}:\radius) 
        arc ({360/\n * (\s - 1)+\margin}:{360/\n * (\s)-\margin}:\radius);

        \def \s {3}
        \node[draw, circle,fill=black!30][minimum size=\nodesize] at ({360/\n * (\s - 1)}:\radius) {$u_i^{s}$};
        \draw[-, >=latex] ({360/\n * (\s - 1)+\margin}:\radius) 
        arc ({360/\n * (\s - 1)+\margin}:{360/\n * (\s)-\margin}:\radius);

        \def \s {4}
        \node[draw, circle][minimum size=\nodesize] at ({360/\n * (\s - 1)}:\radius) {$v^{r}_i$};
        \draw[-, >=latex] ({360/\n * (\s - 1)+\margin}:\radius) 
        arc ({360/\n * (\s - 1)+\margin}:{360/\n * (\s)-\margin}:\radius);
        
        \def \s {5}
        \node[draw, circle,fill=black!30][minimum size=\nodesize] at ({360/\n * (\s - 1)}:\radius) {$u_i^{r}$};
        \draw[-, >=latex] ({360/\n * (\s - 1)+\margin}:\radius) 
        arc ({360/\n * (\s - 1)+\margin}:{360/\n * (\s)-\margin}:\radius);

        \def \s {6}
        \node[draw, circle][minimum size=\nodesize] at ({360/\n * (\s - 1)}:\radius) {$v^{t}_i$};
        \draw[-, >=latex] ({360/\n * (\s - 1)+\margin}:\radius) 
        arc ({360/\n * (\s - 1)+\margin}:{360/\n * (\s)-\margin}:\radius);
        
        \def \s {7}
        \node[draw, circle,fill=black!30][minimum size=\nodesize] at ({360/\n * (\s - 1)}:\radius) {$u_i^{t}$};
        \draw[-, >=latex] ({360/\n * (\s - 1)+\margin}:\radius) 
        arc ({360/\n * (\s - 1)+\margin}:{360/\n * (\s)-\margin}:\radius);

        \def \s {8}
        \node[draw, circle][minimum size=\nodesize] at ({360/\n * (\s - 1)}:\radius) {$v_i^{j}$};
        \draw[-, >=latex] ({360/\n * (\s - 1)+\margin}:\radius) 
        arc ({360/\n * (\s - 1)+\margin}:{360/\n * (\s)-\margin}:\radius);

        \def \xsh {25}
        \def \ysh {0}
        \node[draw, circle,fill=black!30, xshift=\xsh ex, yshift = \ysh ex][minimum size=\nodesize] at (0,0) (c) {$u^{j}_c$};


        \def \xsh {40}
        \def \ysh {0}
        \node[draw, circle,xshift=\xsh ex][minimum size=\nodesize, yshift= \ysh ex] at (0,0) (p1) {$w^{+}_p$};
        \node[draw, circle, right = 4ex of p1][minimum size=\nodesize] (d1) {$v^{+}_d$};
        \node[draw, circle, below = 4ex of p1][minimum size=\nodesize]  (p2) {$w^{-}_p$};
        \node[draw, circle, right = 4ex of p2][minimum size=\nodesize] (d2) {$v^{-}_d$};
        
        \foreach \i in {1,2} {\draw[-] (p\i) -- (d\i);}
        
        

        
      \end{tikzpicture}      
      \par      
    }

    This completes the construction of the voters and their bundles.
    It is straight forward that the construction runs in polynomial time.
    
    By simple calculation, 
    we have that the score difference between candidate~$d$ and $p$ is $4\cdot n$ 
    and the score difference between candidate $c_j$ and $p$ is one.

    Now, we show that the given $(\mathcal{F}, \mathcal{X})$ is a yes instance for \probTTSAT if and only if the constructed instance~$(C, V, \kappa, p, k=2n)$ is a yes instance for \cccdv.

    For the ``only if'' part, let $\sigma\colon \mathcal{X} \to \{0,1\}$ be a satisfying assignment for~$I$.
    Consider some variable~$x_i$.
    Let $f_j$, $f_r$, $j< r$, 
    be the clauses containing~$x_i$ 
    and let $f_s$, $f_t$, $s<t$ be the two clauses containing $\overline{x}_i$.
    Now, if $\sigma(x_i) = 1$, then we add to~$V'$ the two clause voters~$u^j_i$ and $u^r_i$; otherwise we add to~$V'$
    the two clause voters~$u^s_i$ and $u^t_i$.
    We do this for each variable~$x_i$.
    Since each of the bundles of $u^j_i$ and $u^r_i$ contains two $p$-voters 
    and $\sigma$ is a satisfying assignment, 
    one can verify that $|V'| = 2$ and, after deleting the bundles corresponding to $V'$, 
    only two $d$-voters and for each clause~$C_j$ at most two $c_j$-voters remain.
    Since there are in total two $p$-voters, $p$ will co-win with $d$ (and with some clause candidates).

    For the ``if'' part, let $V'$ be set of at most~$2\cdot n$ voters, such that $p$ wins election~$(C, V\setminus \kappa(V'))$.
    First of all, we observe that the score difference between $d$ and $p$ is $4\cdot n$
    and deleting voters will not increase the score of any candidate.
    Thus, in order to make $p$ win, $d$ has to lose at least $4\cdot n$ points.
    Since each bundle contains at most two $d$-voters, 
    it follows that $|V'|=2n$ and for each voter~$v\in V'$, it must be that $\kappa(v)$ contains two $d$-voters.
    This precludes including any voter from the last two groups to $V'$.
    Moreover, 
    for each two voters~$u,v\in V'$ that are on the same cycle (of the bundling graph),
    by the construction of the bundles,
    we must have that $\kappa(u)\cap \kappa(v) = \emptyset$ since otherwise we will not delete enough $d$-voters.
    This implies that, for each variable~$x$, 
    $V'$ contains the two clause voters that correspond to the clauses with either $x$ or $\overline{x}$.
    Now, if it is the case for $x$, then we set $\sigma(x)=1$; otherwise $\sigma(x)=0$.
    The assignment~$\sigma$ is obviously valid.
    It remains to show that $\sigma$ satisfies $\mathcal{F}$.
    Now, observe that each clause candidate~$c_j$ has to lose at least one point because of the original score difference.
    This implies that there must be some literal~$\ell_i$ from $C_j$ such that $u^j_i\in V'$,
    by our assignment we also set $\ell_i$ to $1$, satisfying $C_j$.
    Thus, our constructed~$\sigma$ satisfies all clauses.
\end{proof}
}

If, in addition to the bundles being symmetric and of size at most three, we have only two candidates,
then we can solve \cccav in polynomial time.
First of all, due to these constraints, 
we can assume that the bundling graph~$G_{\kappa}$ is undirected and consists of only cycles and paths.
Then, it is easy to verify that we can consider each cycle and each path separately. 
Finally, we devise a dynamic program for the case when the bundling graph is a path or a cycle,
maximizing the score difference between our preferred candidate~$p$ and the other candidate.
The crucial idea behind the dynamic program is that 
the bundles of a minimum-size solution induce a subgraph where each connected component is small.

\newcommand{\lempathtext}{
  Let $(E=(C, V), W, \kappa, p, k)$ be a \cccav{} instance such that $C=\{p,g\}$, 
and $\kappa$ is symmetric with $G_\kappa$ being a path.
Then, finding a size-at-most-$k$ subset~$W'\subseteq W$ of voters such that the score difference between $p$ and $q$
  in $\kappa(W')$ is maximum
  can be solved in $\ON(|W|^5)$ time, where $|W|$ is the size of the unregistered voter set~$W$.
}
\begin{lemma}
	\label{lemma:path}
	\lempathtext
\end{lemma}
\begin{proof}
  Since $G_\kappa$ is a path, each bundle has at most three voters. 
	We denote the path in $G_\kappa$ by $(w_1,w_2,\ldots, w_{|W|})$ and introduce some definitions for this proof.
	The set $W(s,t) \coloneqq \{w_i \in W\mid s \leq i \leq t \}$ contains all voters on a sequence from $w_s$ to $w_t$. 
	For every subset~$W' \subseteq W$ we define $\gap(W') \coloneqq s_p(\kappa(W')) - s_g(\kappa(W'))$ as the score difference between $p$ and $g$.
	One can observe that if $W'$ is a solution for $(E=(C,V),W,\kappa, p,k)$ then $\gap(W') \geq s_g(V) - s_p(V)$; note that we only have two candidates.
	An \emph{$(s,t)$-proper-subset}~$W'$ is a subset of $W(s,t)$ such that $\kappa(W') \subseteq W(s,t)$.
	A \emph{maximum $(s,t)$-proper-subset}~$W'$ additionally requires that each $(s,t)$-proper-subset $W'' \subseteq W$ with $|W''| = |W'|$ has $\gap(W'') \leq \gap(W')$.

	We provide a dynamic program in which a table entry $T[r,s,t]$ contains a maximum $(s,t)$-proper-subset $W'$ of size $r$.
        We first initialize the table entries for the case where $t-s+1\leq 9$ and $r\le 9$ in linear time.


	For $t-s+1> 9$, we compute the table entry~$T[r,s,t]$ by considering every possible partition of $W(s,t)$ into two disjoint parts. 
	\begin{gather*}
		T[r,s,t] \coloneqq T[r-i,s,s+j] \cup T[i,s+j+1,t],\\
		\text{where } i, j = \argmax_{\substack{0 \leq i \leq r \\ 0 \leq j \leq t - s - 2}} \gap(T[r-i,s,s+j]) + \gap(T[i,s+j+1,t]).
	\end{gather*}
	Note that a maximum $(1,|W|)$-proper-subset $W'$ of size $r-1$ \emph{could} have a higher \gap than a maximum $(1,|W|)$-proper-subset $W''$ of size $r$.

        To show the correctness of our program,
        we define the maximization and minimization function on a set of voters $W'$, which return the largest and smallest index of all voters on the path induced by $W'$, respectively:
	\begin{align*}
		\max(W') \coloneqq \argmax_{i\in |W'|}\{ w_i \in (W') \} \text{ and }
		\min(W') \coloneqq \argmin_{i \in |W'|} \{ w_i \in (W') \}\text{.}
	\end{align*}
	First, we use the following claim to show that 
        each maximum $(s,t)$-proper-subset $W'$ can be partitioned into two $(s,t)$-proper-subsets $W_1,W_2$ such that
        the two sets~$\kappa(W_1)$ and $\kappa(W_2)$ are disjoint.	
        (The formal proof of the following claim can be found in the Appendix.)

	\begin{myclaim}
		\label{claim:split}
		Let $W'$ be a maximum $(s,t)$-proper-subset and $(\max\kappa(W') - \min\kappa(W')+1) > 9$.
		Then, there is a $j$ with $s < j < t$ such that there is an $(s,j)$-proper-subset $W_1$ and a $(j+1,t)$-proper-subset $W_2$ with 
		$|W_1| + |W_2| \leq |W'|$ and $\kappa(W_1 \cup W_2) = \kappa(W')$.
	\end{myclaim}
	\appendixproof{claim:split}{
	\textbf{\Cref{claim:split}.\quad}
	Let $W'$ be a maximum $(s,t)$-proper-subset and $(\max\kappa(W') - \min\kappa(W')+1) > 9$.
	Then, there is a $j$ with $s < j < t$, such that there is a $(s,j)$-proper-subset $W_1$ and a $(j+1,t)$-proper-subset $W_2$ with 
	$|W_1| + |W_2| \leq |W'|$ and $\kappa(W_1 \cup W_2) = \kappa(W')$.
	\begin{proof}
		Let $s' \coloneqq \min \kappa(W')$ and $t' \coloneqq \max \kappa(W')$.
		We split this in two cases.

		First, we consider the case in which $ k \leq 3 $.
				Note that $t' - s' + 1 > 9$ and $s \leq s' \leq t' \leq t$.
				Therefore, $|W(s',t')| > 9$.
				Since the maximum bundle size is $3$, we know that $|\kappa(W')|$ is at most $9$.
				This implies that there is a $w_j \in W(s',t') \setminus \kappa(W')$.
				We set $W_1 \coloneqq \Set{ w_i \in W'}{ i \leq j }$ and $W_2 \coloneqq \Set{ w_i \in W'}{ j < i }$.
				Thus, $|W_1| + |W_2| = |W'|$, $\kappa(W_1 \cup W_2) = \kappa(W')$ and $W_1$ is a $(s,j)$-proper-subset and $W_2$ is a $(j+1,t)$-proper-subset.

		Now, let us consider the case in which $ k > 3 $.
				If the set~$W'$ can be partitioned into two disjoint subsets $W_1$ and $W_2$  where $W_1\cup W_2= W'$, $W_1\cap W_2=W'$, $\kappa(W_1) \cap \kappa(W_2) = \emptyset$ and  $\kappa(W') = \kappa(W_1 \cup W_2)$, we are done.
				So assume there is no such partition.
				That means there is no $w_j \in W(s',t') \setminus \kappa(W')$, otherwise we could split $\kappa(W')$ at position $i$ as we did in case $k \leq 3$.
				Therefore, $W(s',t') = \kappa(W')$.
				Furthermore, we can conclude that there are two bundles of size $2$ (the endpoints of the path) and that all other bundles are of size $3$, because of the maximum bundle size and  $G_\kappa$ being a path.
				At least every second voter on the path is in $W'$, otherwise we could split $\kappa(W')$ at position $i$ or $i+1$, where $w_i,w_{i+1} \not \in W'$.
				Therefore, $|W' \setminus \{ w_{s'}, w_{s'+1}, \dots, w_{s'+8} \}| \leq |W'| - 4$.
				One can observe that $\kappa(\{w_{s'+1},w_{s'+4},w_{s'+7}\}) = \{ w_{s'}, w_{s'+1}, \dots, w_{s'+8} \}$.
				For the set $W'' \coloneqq (W' \setminus \{ w_{s'}, w_{s'+1}, \dots, w_{s'+8} \}) \cup \{w_{s'+1},w_{s'+4},w_{s'+7}\}$ it holds that $|W''| \leq |W'|$ and $\kappa(W') = \kappa(W'')$.
				Furthermore, one can observe that  $w_{s'+4}$ is the only element which has $w_{s'+5}$ in its bundle and $w_{s'+7}$ is the only element which has $w_{s'+6}$ in its bundle.
				That means we can split $W''$ at this point.
				We set $j\coloneqq s'+5$, $W_1 \coloneqq \Set{ w_i \in W'}{ i \leq j }$ and $W_2 \coloneqq \Set{ w_i \in W'}{ j < i }$.
				Thus, $|W_1| + |W_2| \leq |W'|$, $\kappa(W_1 \cup W_2) = \kappa(W')$ and $W_1$ is a $(s,j)$-proper-subset and $W_2$ is a $(j,t)$-proper-subset.
	\qed\end{proof}
	}

	Now, we show that the two subsets~$W_1$ and $W_2$ from \cref{claim:split} are indeed \emph{optimal}:
        There is a $j$ such that $W_1$ is a maximum $(s,j)$-proper-subset and $W_2$ is a maximum $(j+1,t)$-proper-subset.


	Assume towards a contradiction that $W_2$ is a $(j+1,t)$-proper-subset but not a maximum $(j+1,t)$-proper-subset.
	Therefore, there exists a maximum $(j+1,t)$-proper-subset $W_2'$ where $|W_2| = |W_2'|$.
	This implies that $\gap(W_1 \cup W_2') > \gap(W_1 \cup W_2)$.
	This is a contradiction to $W' = W_1 \cup W_2$ being a maximum $(s,t)$-proper-subset.
	The case in which $W_1$ is not a maximum $(s,j)$-proper-subset is analogous.
	
        Altogether, we have shown that we can compute $T[k,s,t]$ in constant time if $t - s + 1 \leq 9$, and that
	otherwise there exist an $i$ and a $j$ such that $T[k,s,t] = T[k-i,s,t-j] \cup T[i,t-j+1,t]$.
	The dynamic program considers all possible $i$ and $j$.
	The table entry $T[i, 1, |W|]$ contains a subset $W' \subseteq W$ of size $i$ with maximum \gap such that $ \kappa(W') \subseteq W(1, |W|) $, which is identical to 
	$ \kappa(W') \subseteq W $.

	This completes the correctness proof of our dynamic program. 
        The table has $\ON(k \cdot |W|^2)$  entries. 
	To compute one entry the dynamic program accesses $\ON(k \cdot |W|)$ other table entries.
	Note that the value $\gap(T[i,s,t])$ can be computed and stored after the entry $T[i,s,t]$ is computed.
	This takes at most $\ON(|W|)$ steps. 
	Thus, the dynamic program runs in $\ON(|W|^5)$~time.
\qed\end{proof}
The dynamic program can be used to solve the same problem on cycles.
\newcommand{\lemcycletext}{
    If \cccav{} has only two candidates~$p,q$ and symmetric bundling function~$\kappa$ with bundles of size at most three
    such that the bundling graph $G_\kappa$ is a cycle,
    then finding a size-at-most-$k$ subset~$W'\subseteq W$ of voters such that the score difference between $p$ and $q$
    in $\kappa(W')$ is maximum
    can be solved in $\ON(|W|^5)$ time, where $|W|$ is the size of the set~$W$ of unregistered voters.
}
\appendixproofwithoutstatement{lemma:cycle}{
  The idea for cycles is similar. 
\begin{lemma}
	\label{lemma:cycle}
	\lemcycletext
\end{lemma}
\begin{proof}
	Let $I = (E=(C,V),W,\kappa, p,k)$ be a \cccav{} instance, where $b=3$, $\kappa$ is symmetric and $C = \{ p, g \}$.
	Let $w_1\dots w_{|W|}w_1$ be the cycle in $G_\kappa$.
	Suppose that $|W| \leq 9$.
	We compute $s_p(V)$ and $s_g(V)$ in linear time.
	Every $W' \subseteq W$ of size at most $k\leq |W| \leq 9$ is a solution if $\kappa(W')$ contains $s_g(V) - s_p(V)$ more $p$-voter than $g$-voter.
	Since $|W|$ is upper bounded by a constant, the size of the power set of $W$ is also upper bounded by a constant.
	Therefore, we can check in constant time whether one of the subsets of $W$ is a solution.

	Now let us consider the interesting case. 
	Suppose that $|W| > 9$.
	The idea is to break the cycle into a path and solve it with \cref{lemma:path}.
	We create nine \cccav{} instances $I_1,\dots,I_9$ where 
	$I_i \coloneqq (E=(C,V),W \cup \{ w_b, w_e \},\kappa_i, p,k)$, 
	and $w_b$ and $w_e$ are $g$-voters and
	\begin{align*}
		\kappa_i : W \rightarrow 2^W, w \mapsto 
		\begin{cases}
			\{ w_b, w_i \},					& \text{if } w = w_b \\
			\{ w_e, w_{i+1} \}, 				& \text{if } w = w_e \\
			(\kappa(w) \setminus \{ w_{i+1} \}) \cup \{ w_e \}, & \text{if } w = w_i \\
			(\kappa(w) \setminus \{ w_i \}) \cup \{w_b \},	& \text{if } w = w_{i+1} \\
			\kappa(w),					& \text{otherwise,}
		\end{cases}
	\end{align*}
	where $i \in [9]$.
	\begin{figure}[t]
		\centering
\hspace*{-9.5em}
\raisebox{-3em}{
\begin{tikzpicture}[scale=0.8]
	\tikzset{vertex/.style = {shape=circle,draw,minimum size=1.5em}}
	\node[vertex] (1) at (0,0) {$w_1$};
	\node[vertex] (2) at (1,0) {$w_2$};
	\node[vertex] (3) at (2,0) {$w_3$};
	\node[vertex] (4) at (3,0) {$w_4$};
	\node[vertex] (5) at (4,0) {$w_5$};
	\node[vertex] (b) at (5,0) {$w_b$};

	\node[vertex] (e) at (7,0) {$w_e$};
	\node[vertex] (6) at (8,0) {$w_6$};
	\node[vertex] (7) at (9,0) {$w_7$};
	\node[vertex] (8) at (10,0) {$w_8$};
	\node[vertex] (9) at (11,0) {$w_9$};
	\node (dots) at (12,0) {$\dots$};
	\draw (1) to (2);
	\draw (2) to (3);
	\draw (3) to (4);
	\draw (4) to (5);
	\draw (5) to (b);

	\draw (e) to (6);
	\draw (6) to (7);
	\draw (7) to (8);
	\draw (8) to (9);
	\draw (9) to (dots);
	\draw (dots) to [bend left,in=15,out=165](1);

\end{tikzpicture}}

		\caption{Part of the construction used in \cref{lemma:cycle}.
		Specifically, it shows how the cycle is reduced to a path in $I_5$.}
		\label{fig:cycle}
	\end{figure}
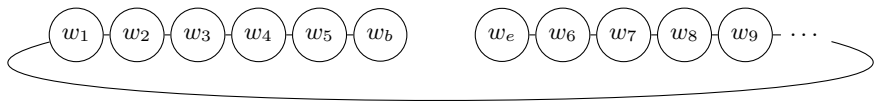
	The bundling function $\kappa_i$ is modified in such a way that the instance $I_i$ becomes a path and has $g$-voter on its endpoints, depicted in \cref{fig:cycle}. 
	Therefore, we can solve all 9 instances $I_i$ in $\ON(|W|^5)$, 
        where $|W|$ is the size of the unregistered voters (see \cref{lemma:path}).

	It remains to be shown that $I$ is a yes-instance if and only there it exists an $i \in [9]$ such that $I_i$ is a yes-instance.

	For the ``only if'' part,
			let $I_i$, for some $i \in [9]$, be a yes-instance.
			Thus, there is a set $W' \subseteq W \cup \{ w_b, w_e \}$ of size at most $k$ such that $p \in$ \plur$(C,V \cup \kappa_i(W'))$.
			One can observe that the bundles of $w_b$ and $w_e$ are useless, because they each have only two voters and at least one of them is a $g$-voter.
			For this reason we assume without loss of generality that $W' \subseteq W$.
			Since we replaced an edge between $w_i$ and $w_{i+1}$ with two $g$-voters $w_b$ and $w_e$ that are not connected, we can conclude that $w_b \in \kappa_i(W') \aq w_i \in W'$ and $w_e \in \kappa_i(W') \aq w_{i+1} \in W'$.
			We now check what happens when $w_i$ or $w_{i+1}$ is in $W'$.
			\begin{compactenum}[(i)]
				\item 
					Assume $w_i,w_{i+1} \in W'$.
					Then $\kappa_i(W') \setminus \{ w_b, w_e \} = \kappa(W')$.
					The set $\kappa(W')$ has $2$ $g$-voter less than $\kappa_i(W')$.
					Thus, $p \in$ \plur$(C,V \cup \kappa(W'))$.
				\item
					Assume $w_i \in W'$ and $w_{i+1} \not \in W'$.
					Then $\kappa(W') = (\kappa_i(W') \setminus \{ w_b \}) \cup \{ w_{i+1} \}$.
					In contrast to $\kappa_i(W')$, $\kappa(W')$ loses the $g$-voter $w_b$ and may win the voter $w_{i+1}$.
					It does not matter whether $w_{i+1}$ is a $p$-voter or $g$-voter, as
					$ s_g(\kappa(W')) \le s_g(\kappa_i(W'))$ 
					and $ s_p(\kappa(W')) \ge s_p(\kappa_i(W'))$.
					Thus, $p \in$ \plur$(C,V \cup \kappa(W'))$.
				\item
					Assume $w_{i+1} \in W'$ and $w_i \not \in W'$.
					This case is analogous to (ii).
			\end{compactenum}
			In the case where neither $w_i$ nor $w_{i+1}$ is in $W'$, we immediately get $\kappa_i(W') = \kappa(W')$ from the definition of $\kappa_i$.
			Therefore, $I$ is a yes-instance.

		For the ``if'' part, 
			let $I$ be a yes-instance.
			Thus, there is a set $W' \subseteq W$ of size at most $k$ such that $p \in$ \plur$(C,V \cup \kappa(W'))$.
			It is sufficient to show that for some $i$, $1 \le i \le 9$ it holds that $w_i,w_{i+1} \not \in W'$,
			because these are the only cases in which $\kappa_i$ differs from $\kappa$.
			Thus, $\kappa(W') = \kappa_i(W')$, which means that $I_i$ is a yes-instance.

			Now, assume that such an $i$ does not exist.
			Then, on the path $w_1\dots w_9$ at least every second voter is in $W'$.
			Therefore, $|W' \setminus \{ w_1, \dots, w_9 \}| \leq |W'| - 4$.
			One can observe that $\kappa(\{w_2,w_5,w_8\}) = \{ w_1, \dots, w_9 \}$.
			For the set $W'' \coloneqq (W' \setminus \{ w_1, \dots, w_9 \}) \cup \{w_2, w_5, w_8\}$ it holds that $|W''| \leq |W'|$ and $\kappa(W') = \kappa(W'')$, where $|W''| \leq |W'|$.
			Furthermore, one can observe that $w_6,w_7 \not \in W''$.
			That means $\kappa(W'') = \kappa_5(W'')$ and hence $p \in$ \plur$(C,V \cup \kappa_5(W''))$.
			This implies that $I_5$ is a yes-instance.
\qed\end{proof}
}
Altogether, we obtain the following.
\begin{theorem}
	\label{thm:m2symm}
	\cccav{} with a symmetric bundling function, maximum bundle size of three, and for two candidates can be solved in $\ON(|W|^5)$ time, 
        where $|W|$ is the size of the unregistered voter set.
	\label{thm:algo1}
\end{theorem}
\begin{proof}
	Let $(E=(C,V),W,\kappa, p,k)$ be a \cccav{} instance, where the maximum bundle size $ b $ is three, $\kappa$ is symmetric, and $C = \{ p, g \}$.
        This means that all connected components $C_1,\dots,C_{\ell}$ of $G_\kappa$ are path or cycles.
	Furthermore, all bundles only contain voters from one connected component.
	We define a dynamic program in which each table entry $A[i,s,t]$ contains a solution~$W' \subseteq W$ of size $i$,
        where $\kappa(W') \subseteq V(C_s) \cup \dots \cup V(C_t)$ and $1 \leq s \leq t \leq \ell$:
	\begin{compactenum}[(i)]
          \item If $s=t=j$, then $A[i,s,t] = T[i,1,|V(C_j)|]$, where $T$ is the dynamic program of $C_j$, depending on whether $C_j$ is a path or cycle. 
          \item Otherwise, we build the table as follows:
          
          {\centering
            $A[d,s,t] = A[d - i, s, s+j] \cup A[i, s+j+1, t], $ where 

            $\displaystyle i,j = \argmax_{\substack{0 \leq i \leq d \\ 1 \leq j \leq t-s-1}} \gap(A[d-i,s,s+j]) + \gap(A[i,s+j+1,t])$.
          \par}
        \end{compactenum}
        \noindent  Each of the table entries~$A[i,j,j]$ can be computed in $\ON(i^2 \cdot |V(C_j)|^3)$ time (see \cref{lemma:path,lemma:cycle})
        and each of the table entries~$A[i,s,t]$ for $s< t$ can be computed in $O(k\cdot \ell)$ time.
        Since we have $k\cdot \ell^2$ entries, the total running time is
        \smallskip

        {\centering
        $\sum_{i=1}^{\ell} \ON(k^2 \cdot |V(C_j)|^3) = \ON(k^2) \sum_{i=1}^{\ell} \ON(|V(C_i)|) = \ON(k^2 \cdot |W|^3)\text{.}$\hfill\qed
      \par}
\end{proof}

From the polynomial-time solvability of \cref{thm:m2symm} and by \cref{thm:tur}, 
we obtain the following results:
\begin{corollary}
	\label{thm:cons_des_poly}
	\cccdv{} with a symmetric bundling function, 
        a maximum bundle size of three and two candidates can be solved in $ \ON(|V|^5) $ time, where $|V|$ is the size of the voters.
\end{corollary}

\newcommand{\cordespolytext}{%
  	\cdcav{} and \cdcdv{} with a symmetric bundling function and maximum bundle size three can be solved in time~$ \ON(m\cdot |W|^5) $ and $\ON(m\cdot |V|^5)$, respectively, where $ m $ is the number of candidates,
       and $|W|$ and $|V|$ are the sizes of the unregistered and registered voter set, respectively.%
}
\begin{corollary}
	\label{thm:des_poly}
	\cordespolytext
\end{corollary}
\appendixproofwtext{thm:des_poly}{\cordespolytext}{
\begin{proof}
	To obtain the result for \cdcav{}, use the Turing reduction from \cdcav{} to \cccav{} (\cref{thm:tur}).
	Note that the instances produced by the Turing reduction only have two candidates.
	Thus, we can use \cref{thm:m2symm}.

	Analogously, we can use \cref{thm:cons_des_poly} to obtain the result for \cdcdv{}.

	As in both Turing reductions we create $O(m)$ instances, 
        we need to solve $ \ON(m) $ instances, resulting in a running time of $ \ON(m\cdot n^5)$ in both cases, where $n$ is the number of either the unregistered (for \cdcav) or the registered (for \cdcdv) voters .
\qed\end{proof}
}

\section{Controlling Voters with Disjoint Bundles} 
\label{sec:disjoint}
We have seen in \cref{sec:des_vs_cons} that 
the interaction between the bundles influences the computational complexity of our combinatorial voter control problems.
For instance, adding a voter~$v$ to the election may lead to adding another voter~$v'$ with $v\in \kappa(v)$.
This is crucial for the reductions used to prove \cref{thm:similar_summary} and \cref{thm:cons_hard}.
Thus, it would be interesting to know whether the problem becomes tractable if it is not necessary to add two bundles that share some voter(s).
More specifically, we are interested in the case where the bundles are disjoint,
meaning that we do not need to consider every single voter, but only the bundles as a whole, 
as it does not matter which voters of a bundle we select.


First, we consider disjoint bundles of size at most two.
This is the case for voters who have a partner. 
If a voter is convinced to participate in or leaves the election, 
then the partner is convinced to do the same.
Note that this is equivalent to having symmetric bundles of size at most two.
\citet[Theorem 6]{bulteau2015combinatorial} constructed a linear-time algorithm for \cccav{} if the maximum bundle size is two and $ \kappa $ is 
a full-$ d $ bundling function (which implies symmetry). 
We can verify that their algorithm actually works for disjoint bundles of size at most two.
Thus, we obtain the following.

\newcommand{\btwosymm}{%
  \cccav{} with a symmetric bundling function and with bundles of size at most two
  can be solved in $ \ON(|I|) $ time, where $|I|$ is the input size. %
}
\begin{observation}\label{thm:b2symm}
  \btwosymm
\end{observation}

If we want to delete instead of add voter bundles, 
the problem reduces to finding a special variant of the \textsc{$f$-Factor} problem, 
which is a generalization of the well-known matching problem and can still be solved in polynomial time~\cite{Anstee1985,Anstee1993}.

\newcommand{\btwosymmdel}{%
  \cccdv{} with a symmetric bundling function and with bundles of size at most two
  can be solved in polynomial time.
}
\begin{theorem}\label{thm:b2symm-del}
  \btwosymmdel
\end{theorem}

\appendixproofwtext{thm:b2symm-del}{\btwosymmdel}{
  \begin{proof}
    Let $I=((C,V), \kappa, p, k)$ be a \cccdv instance, 
    where the bundling function~$\kappa$ is symmetric and each bundle has at most two voters.
    We first consider the case where each bundle has size exactly two.
    Since deleting bundles never increases any candidate's score,
    we can assume that we do not delete any bundle that contains a $p$-voter.
    Under this assumption, the original score of $p$ equals her score in the final election.
    Furthermore, if we know the score of $p$ in the final election, 
    then we know for each remaining candidate~$c_i\neq p$ how many $c_i$-voters we need to delete to make her have no more score than $p$; let $d_i$ be the score difference between $c_i$ and $p$.
    Due to this observation, 
    we can construct a multigraph (with loops),
    which contains a vertex~$u_i$ for each candidate~$c_i$ except $p$,
    and where for each bundle with two voters supporting $c_i$ and $c_j$
    there is an edge incident to the vertices~$v_i$ and $v_j$
    (note that if $c_i=c_j$, then the edge is a loop).
    Now, deleting minimum bundles to make each candidate~$c_i$ lose at least $d_i$ points is equivalent to 
    finding a subgraph with minimum number of edges where each vertex has degree~$d_i$.
   \citet{Anstee1993} showed that the latter problem can be solved in polynomial time.

   To handle the case where some bundle has only one voter,
   we first observe that if our preferred candidate~$p$ has zero points, 
   then we need to delete all bundles.
   Note that the underlying bundling graph consist of only disjoint edges and isolated vertices.
   The minimum number of bundles to be deleted equals the number of edges plus the number of isolated vertices.
   
   Now, we consider the case where our preferred candidate~$p$ has at least one point.
   For each bundle with only one voter~$v_i$, 
   we introduce a dummy candidate~$d_i$ and a dummy $d_i$-voter,
   and bundle the voter with $v_i$ together.
   Since each of the dummy candidates has exactly one point,
   the original instance (with $p$ having at least one point) is a yes instance if and only if the modified instance is a yes instance.
   In the modified instance, every bundle has size exactly two.
   The first case applies.
  \end{proof}
}

If we drop the restriction on the bundle sizes but still require the bundles to be disjoint, 
then \cccav{} and \cccdv{} become parameterized intractable with respect to the solution size. 

\begin{theorem}
	\label{thm:disjoint_hard}
	Parameterized by the solution size~$k$,
        \cccav{} and \cccdv{} are \Wone{}-hard and \Wtwo{}-hard respectively, even for disjoint bundles.
\end{theorem}

\begin{proof}[with only the construction for the \Wone{}-hardness proof of \cccav{}]
	We provide a parameterized reduction from the \Wone{}-complete problem
	\indset{} (parameterized by the ``solution size'') which, given an undirected graph $ G = (V(G), E(G)) $ 
	and a natural number $ h \in \N $, asks whether $G$ admits a size-$ h $ \emph{independent set}~$U\subseteq V(G)$, 
	that is, all vertices in $U$ are pairwise non-adjacent.
	Let $ (G, h) $ be an \indset{} instance with $ E(G) = \{ e_1, \dots,$ $e_{m-1} \} $ and $V(G) = \{ u_1, \dots, u_n \} $. 
	Without loss of generality, we assume that  $ G $ is connected and $ h \ge 3 $.
        We construct an election $ E = (C, V) $
	with candidate set~$C \coloneqq \{ p \} \cup \{g_j \mid e_j \in E(G)\}$.
	For each edge~$e_j\in E$, we construct $h-1$ registered voters that all have $g_j$ as their favorite candidate.
        In total, $V$ consists of $(h-1)\cdot (m-1)$ voters.
       
	The unregistered voter set $ W $ is constructed as follows:
	For each vertex~$u_i \in V(G)$, add a $ p $-voter $ p_i $, 
	and for each edge~$e_j$ incident with $ u_i $, add a $ g_j $-voter $ a_j^{(i)} $.
	The voters constructed for each vertex $ u_i $ are bundled by the bundling function~$\kappa$.
	More formally, for each~$u_i \in V(G)$ and each $e_j\in E(G)$ with $u_i \in e_j$, it holds that 

        {\centering
		$  \kappa(p_i) =  \kappa(a_j^{(i)}) \coloneqq \{ p_i \} \cup \{a_{j'}^{(i)}\mid e_{j'} \in E(G) \land u_i \in e_{j'}\}.$
	\par}
      
      \noindent To finalize the construction, we set $ k \coloneqq h $.
	The construction is both a polynomial-time and a parameterized reduction, 
	and all bundles are disjoint. 
        To show the correctness, 
        we note that $p$ can only win if only if her score can be increased to at least $h$ without giving any other candidate more than one more point.
        The solution corresponds to exactly to a subset of $h$ vertices that are pairwise non-adjacent.
	The detailed correctness proof and the remaining proof for the \Wtwo-hardness result can be found in the Appendix. \qed
 \end{proof}

\appendixproof{thm:disjoint_hard}{
In order to prove \cref{thm:disjoint_hard},
we first show \Wone{}-hardness of \cccav{} (\cref{lem:disjoint_hard_cccav})
and then \Wtwo{}-hardness of \cccdv{} (\cref{lem:disjoint_hard_cccdv}).

\begin{lemma}
\label{lem:disjoint_hard_cccav}
	\cccav{} is \Wone{}-hard, even for disjoint bundling function.
\end{lemma}
\begin{proof}
	We provide a parameterized reduction from the \Wone{}-complete problem
	\indset{} (parameterized by the ``solution size''), defined as follows:
        \probDef{\indset}
        {An undirected Graph $ G = (V(G), E(G)) $ 
          and a natural number $ h \in \N $.}
        {Does $ G $ admit a size-$ h $ \emph{independent set}, 
          that is, all vertices in $U$ are pairwise non-adjacent.?}

	Let $ (G, h) $ be an \indset{} instance with $ E(G) = \{ e_1, \dots,$ $e_{m-1} \} $ and $V(G) = \{ u_1, \dots, u_n \} $. 
	Without loss of generality, we assume that  $ G $ is connected and $ h \ge 3 $.
        We construct an election $ E = (C, V) $
	with candidate set~$C \coloneqq \{ p \} \cup \Set{g_j}{e_j \in E(G)}$.
	For each edge~$e_j\in E$, we construct $h-1$ registered voters that all prefer $g_j$ most.
        In total, $V$ consists of $(h-1)(m-1)$ voters.
       
	The unregistered voter set $ W $ is constructed as follows:
	For each vertex $ u_i \in V(G) $, add a $ p $-voter $ p_i $, 
	and for each edge~$e_j$ incident with $ u_i $, add a $ g_j $-voter $ a_j^{(i)} $.
	The voters constructed for each vertex $ u_i $ are bundled by $ \kappa $.
	More formally, for each~$u_i \in V(G)$ and each $e_j\in E(G)$ with $u_i \in e_j$, it holds that 
        \begin{align*}
		\kappa(p_i) =  \kappa(a_j^{(i)}) \coloneqq \{ p_i \} \cup \{a_j^{(i)}\mid e_j \in E(G) \land u_i \in e_j\}.
	\end{align*}
        To finalize the construction, we set $ k \coloneqq h $.
	
        The construction is both a polynomial-time and a parameterized reduction, 
	and all bundles are disjoint. 
        To show the correctness, 
        we note that $p$ can only win if only if her score can be increased to at least $h$ without giving any other candidate more than one more point.
        The solution corresponds to exactly to a subset of $h$ vertices that are pairwise non-adjacent.

	Now, it remains to show that $ G $ has a size-$ h $ independent set if and only if 
	$ ((C, V), W, \kappa, p, k) $ is a yes-instance for \cccav{}.

	For the ``only if'' part, suppose that $ U \subseteq V(G) $ is a size-$h$
	independent set for $ G $. 
	Define the subset $ W' $ as the voters $ p_i \in W $ corresponding to the vertices
	$ u_i \in U $.
	Obviously, $ |W'| = h = k $. 
	As per definition, $ U $ does not contain adjacent vertices, 
	each candidate $ g_j $ may achieve a score increase of at most one, while 
	$ p $ achieves a score increase of $ k $. 
	As the initial difference in scores between $ p $ and every candidate $ g_j $
	is $ k - 1 $, $ p $ co-wins the final election with 
	each candidate~$g_j$ such that $u_i \in e_j$ for some voter~$p_i \in W'$.

	For the ``if'' part, suppose that there is a subset $ W' $ of size at most~$k$ 
	such that $ p $ is a winner of the Plurality election $ (C, V \cup \kappa(W')) $.
        First, we claim that $|W'|=k$ and $\kappa(W')$ contains exactly $k$ $p$-voters:
        Since the original score difference between $p$ and any other candidate~$g_j$ is $k-1$
        and since each unregistered voter's bundle contains exactly one $p$-voter and at least one non-$p$-voter,
        it follows that in the final election at least one non-$p$-candidate has score of at least $k$, but $p$ can have a score increase of at most $k$.
        This means 
        \begin{enumerate}[(1)]
          \item\label{sol=k} that $p$'s final score must increase to $k$, that is, $\kappa(W')$ must have exactly $k$ $p$-voters,
          and 
          \item \label{scoreinc<2} that no other candidate can have a score increase of more than one.
        \end{enumerate}
     
        Now, define $U\coloneqq \{u_i \mid p_i \in \kappa(W')\}$. Due Observation~\eqref{sol=k} we have that $|U|=k=h$.
        To show that $U$ is also an independent set,
        we consider two arbitrary vertices~$u_i, u_{\ell}\in U$. 
        Suppose for the sake of contradiction that $\{u_i,u_{\ell}\}\in E(G)$, denote this edge by $e_j$.
        By the construction of our bundling function, it must hold that $a^{(i)}_j, a^{(\ell)}_j \in \kappa(W')$.
        Then, $g_j$ would have a score increase of at least two---a contradiction to our observation~\eqref{scoreinc<2}.
\qed\end{proof}

\begin{lemma}
\label{lem:disjoint_hard_cccdv}
	\cccdv{} is \Wtwo{}-hard, even for disjoint bundling function.
\end{lemma}
\begin{proof}
	We provide a parameterized reduction from the \Wtwo{}-complete problem \domset{}
	(see \cref{app:w2} for the corresponding definition).

	Let $ (G = (V(G), E(G)), h) $ be a \domset{} instance with $ E(G) = \{ e_1, \dots,$ $e_{m'} \} $ and $V(G) = \{ u_1, \dots, u_{n'} \} $. 
	Let $ \Delta_{max} $ denote the maximum degree of $ G $.
        We construct an election $ E = (C, V) $ with candidate set~$C = \{ p, g_1, \dots, g_{n'} \} $.
        
	The voter set $ V $ consists of two groups:
        \begin{enumerate}
          \item For each vertex $ u_i \in V(G) $, add a $ g_i $-voter, denote as \emph{vertex voter}~$v_i^{(i)} $,
          and for each of its neighbors~$ u_j \in N(u_i)$ 
          add a$ g_j $-voter, denote as \emph{neighbor voter}~$v_j^{(i)}$. 
          Note that, for each vertex~$u_i$, we have added exactly~$|N[u_i]|$ voters that all most prefer~$g_{i}$.
          The vertex voters and neighbor voters constructed due to vertex $u_i$ are contained in one bundle. 
          Formally, the bundling function for these voters looks as follows.
          \begin{align*}
            \forall u_i \in V(G)\colon & \kappa(v_i^{(i)}) \coloneqq \{ v_j^{(i)} \mid u_j \in N[u_i] \}\text{, and } \\
                                       &		\forall u_j \in N(u_i)\colon \kappa(v_j^{(i)}) \coloneqq \kappa(v_i^{(i)})
          \end{align*}
          \item For each vertex~$ u_i \in V(G) $, add $ \Delta_{max} + 1 - |N[u_i]| $ $ g_i $-voters. 
          Furthermore, add $ \Delta_{max} $ $ p $-voters to $ V $.
          All these voters are bundled together. 
          Note that in this bundle, each candidate~$g_i$ has exactly $|N[u_i]|$ points less than $p$.
        \end{enumerate}
        We use $D$ to denote the set of all voters constructed in the second group.
	Finalizing the construction, set $ k \coloneqq h $.

	The construction is both a polynomial-time and a parameterized reduction, 
	and the bundling function is disjoint. 
	Note that for all candidates $ g_i $ have score~$ s_{g_i}(V) = \Delta_{max} + 1 $, and 
	the difference in score between every $ g_i $ and the candidate $ p $ is
	$ s_{g_i}(V) - s_p(V) = \Delta_{max} + 1 - \Delta_{max} = 1 $.

	It remains to show that there is a dominating set of size at most $ h $
	if and only if there is a subset $ V' $ of voters of size at most $ k $, 
	such that if their respective bundles are deleted from the election, $ p $
	becomes a Plurality winner of the election.

	For the ``only if'' part, given a dominating set $ U $ of size at most $ h $, 
	we define $ V' $ to be the corresponding voter set, that is, 
	$ V' \coloneqq \{v_i^{(i)}\mid u_i \in U\} $. 
	It is easy to verify that $ | V' | \le h = k $ and $ p $ has a score of 
	$ \Delta_{max} $, while all other candidates have a score of at most $ \Delta_{max} $.
	Thus, $ p $ is a Plurality-winner of the election
	$ (C, V \setminus \kappa(V')) $.

	For the ``if'' part, suppose that there is a subset $ V' \subseteq V $ 
	of size at most $ k $ such that $ p $ is a winner of the Plurality election 
	$ ( C, V \setminus \kappa(V') ) $.
	First of all,  since all voters in $D$ are bundled together such that $p$ has more supporters than any other candidate in this bundle,
        by the disjoint property of the bundles, we know that  $p$ will also win the election~$(C, V\setminus \kappa(V'\setminus D)$.
        Thus,  we can assume that $V'$ does not contain any voter from the second group~$D$
        This means that in the final election, $p$ will have score $\Delta_{\max}$,
        implying that each  $g_i$ has to lose at least one point.
        In other words,
        \begin{align*}&\kappa(V') \text{ contains at least one } g_i\text{-voter for each candidate~} g_i. & (*)\end{align*}

        Now, we define the vertex subset~$U \coloneqq \{u_i \mid u^{(i)}_{i} \in \kappa(V')\}$ and show that it is a dominating set of size at most $h$.
        Since all bundles are disjoint and $|V'|\le k$, it is clear that $|U|\le k=h$.
        To show that $U$ is a dominating set, 
        we consider an arbitrary vertex~$u_j \notin U$.
        By $(*)$, we know that $\kappa(V')$ contains a $g_j$-voter; note that he is from the first group.
        Let this voter be $v^{(i)}_j$.
        Due to the construction of our bundling function,
        it follows that $\kappa(V')$ also contains $v^{(i)}_i$,
        implying that $U$ contains $u_i$.
        Thus $u_j$ is dominated by $u_i\in U$.
\qed\end{proof}
}

For destructive control, it is sufficient to guess a potential defeater~$d$ out of $m-1$ possible candidates
that will have a higher score than $p$ in the final election
and use a greedy strategy similar to the one used for \cref{thm:b2symm}
to obtain the following result.
	
\newcommand{\btwosymmdes}{%
  \cdcav{} and \cdcdv{} with a symmetric bundling function and disjoint bundles
  can be solved in $ \ON(m\cdot |I|) $ time, where $|I|$ is the input size and $m$ the number
  of candidates.%
}

\begin{theorem}
	\label{thm:symmdes}
        \btwosymmdes
\end{theorem}

\appendixproofwtext{thm:symmdes}{\btwosymmdes}{
\begin{proof}

  We consider \cdcav{} first. 
  Let $I=(E=(C,V),W,\kappa, p,k)$ be a \cdcav{} instance with $k$ being symmetric.
  To make $p$ lose, it is enough to add voters such that there is some other candidate~$c$ that has a higher score than $p$ in the final election.
  Due to this observation, we ``guess'' the defeater~$c\in C\setminus \{p\}$
  and we use an even simpler greedy strategy than the one used for \cref{thm:b2symm} (see Theorem~6 by \citet{bulteau2015combinatorial})
  to add voters to maximize the score difference between $c$ and $p$:
%
  Since the bundles are disjoint, we greedily add bundles that maximally improve the
  score difference between~$c$ and~$p$ (in favor of $c$).
  If using up $k$ bundles, we can make $c$ have at least the same score as $p$,
  then return yes; otherwise we continue with the next defeater.
  It is easy to verify that this approach is correct and the running time is $O(m\cdot |I|)$.

  The strategy for \cdcdv{} works analogously.
\end{proof}
}

\section{Controlling Voters with Unlimited Budget} 
\label{sec:inapproximability}
To analyze election control, 
it is interesting to know whether a solution exist at all, without bounding its size.
Indeed, \citet{bartholdi1992hard} already considered the case of unlimited solution size for the constructive candidate control problem.
They showed that the problem is already $\NP$-hard, even if the solution size is not bounded.
(The non-combinatorial destructive control by adding unlimited amount of candidates 
is shown to be also $\NP$-hard by \citet{hemaspaandra2007anyone}.)
In contrast, the non-combinatorial voter control variants
are linear-time solvable via simple greedy algorithms~\cite{bartholdi1992hard}. 
This leads to the question whether the combinatorial structure increases the complexity.
To this end, we relax the four problem variants so that the solution can be of arbitrary size
and call these problems
\cccauv{}, 
\cdcauv{}, 
\cccduv{} and
\cdcduv{}.


First of all, we observe that \cccduv{} becomes trivial
if no unique winner is required.

\newcommand{\consdeleasy}{
  Let $I = (E=(C,V),\kappa,p)$ be a \cccduv{} instance.
  Then $I$ is a yes-instance.%
}
\begin{lemma}\label{lem:consdeleasy}
\consdeleasy
\end{lemma}
\appendixproofwtext{lem:consdeleasy}{\consdeleasy}{
\begin{proof}
Let $I = (E=(C,V),\kappa,p,k)$ be a \cccduv{} instance.
Set $V' \coloneqq V$.
Since there is no voter in the election $(C,V\setminus \kappa(V'))$, 
all candidates have the same score.
Thus $p \in$ \plur$(C, V\setminus V')$.
\qed\end{proof}
}

If we consider a voting rule $ \mathcal{R} $ that only returns unique winners, 
then \cdcduv{} also becomes tractable since we only need to delete all voters. 

For the constructive adding voters case, we obtain $\NP$-hardness.
The idea for the reduction derives from the \Wone{}-hardness proof of \cccav{} shown by
\citet{bulteau2015combinatorial}.

\newcommand{\unlimitedhard}{\cccauv{} is $\NP$-hard.}
\begin{lemma}\label{lem:limitedhard}
  \unlimitedhard
\end{lemma}

\appendixproofwtext{lem:limitedhard}{\unlimitedhard}{
\begin{proof}
	To prove $\NP$-hardness, we extend the proof idea for \Wone{}-hardness from \citet{bulteau2015combinatorial}.
	Let $(G=(V(G),E(G)),h)$ be a \clique{} instance. 
	(The definition of \clique{} can be found in \cref{app:w1}.)
	Without loss of generality, we assume $ h \ge 4 $. 
	We construct an election $E\coloneqq(C,V)$ where $C\coloneqq\{p,g,x\}$. 
	$V$ consists of ${h \choose 2}$ $p$-voters and $(2 \cdot {h \choose 2} - h)$ $g$-voters.
	The set of unregistered voters $W$ is composed as follows:
	\begin{compactitem}[-]
		\item For each vertex $u \in V(G)$, add a $g$-voter $w_u \in W$ with $\kappa(w_u) = \{ w_u \}$.
			We call $w_u$ a \emph{vertex-voter}.
		\item For each edge $e = \{ u,u'\} \in E(G)$ we add a $p$-voter $w_e \in W$ 
			and two $x$-voters $w_e',w_e'' \in W$ such that 
			$\kappa(w_e) = \{ w_e, w_u, w_{u'}, w_e', w_e'' \}$, 
			$\kappa(w_e') = \{ w_e'\}$ 
			and $\kappa(w_e'') = \{ w_e'' \}$.
			We call $w_e$ an \emph{edge-voter}.
	\end{compactitem}
	Obviously, our construction is a polynomial reduction. 
	It remains to show that there is a clique of size at least $ h $ if and only if 
	there is a subset $ W' $ such that $ p $ is a winner of the Plurality election
	$ (C, V \cup \kappa(W')) $.

	For the ``if'' part, suppose that there is a subset $ W' \subseteq W $ such that 
	$ p \in \text{\plur}(C, V \cup \kappa(W')) $.
	We show that the vertex set $ U \coloneqq \Set{u\in V(G)}{w_e \in W' \land u \in e} $ is a 
	clique of size $ h $ in $ G $.
	First, we observe that a solution $W' \subseteq W$ cannot have more than 
	$h \choose 2$ edge-voters, otherwise $x$ achieves a higher score than $p$.
	Second, a solution $W'$ must have $h \choose 2$ edge-voters which have only
	$h$~different vertex-voters in its bundles (a clique in $G$).
	For a precise argumentation see the proof of \cref{thm:w1}.


%
%
%

	For the ``only if'' part, suppose that $ U \subseteq V(G) $ is a size-$ h $ clique for $ G $.
	We construct the subset $ W' $ by adding any edge voter $ w_e $ with $ e \in E(G[U']) $. 
	Now it is easy to check that $ p \in \text{\plur} (C, V \cup \kappa(W'))$.
	(Compare with the proof of \cref{thm:w1}.)
\qed\end{proof}
}

\cref{lem:limitedhard} immediately implies the following inapproximability result for the optimization
variant of \cccav (denoted as \mcccav{}), aiming at minimizing the solution size. 
\newcommand{\thminapproxtext}{There is no polynomial-time approximation algorithm for \mcccav{}, unless $\PP=\NP$.}
 \begin{theorem}\label{thm:inapproximability}
\thminapproxtext
\end{theorem}

\appendixproofwtext{thm:inapproximability}{\thminapproxtext}{
\begin{proof}
Assume towards a contradiction that there is an approximation algorithm 
$A_\alpha$ for \mcccav{}
which runs in polynomial time and provides a solution $W_\alpha$ such that 
$\alpha \cdot |OPT| \geq |W_\alpha|$, where $OPT$ is a solution of minimum size.
Let $(E,W,\kappa,p)$ be an \cccauv{} instance.
One can create a \mcccav{} instance $(E,W,\kappa,p)$ and compute a solution 
$W_\alpha$ of size $\alpha \cdot |OPT|$ in polynomial time.
$W_\alpha$ is a solution for the instance $(E,W,\kappa,p)$ of the $\NP$-hard problem \cccauv{}.
This is a contradiction unless $ \PP = \NP $.
\qed\end{proof}
}

\section{Conclusion}
\label{sec:conclusion}



We extend the study of combinatorial voter control problems introduced by \citet{bulteau2015combinatorial} and obtain that the destructive control variants tend to be computationally easier than their constructive cousins.

Our research leads to several open questions and further research opportunities.
First,  we have shown hardness results for the adding candidate case:
if the bundling function consists of disjoint cliques, 
then parameterized by the solution size, \cccav{} is \Wone{}-hard and \cdcav{} is \Wtwo-hard. 
If one could also determine the complexity upper bound, that is, under the given restrictions,
if \cccav{} would be contained in \Wone{}, then this would yield another difference in complexity between the 
destructive and the constructive variants. 
This also leads to the question whether the problem variants in their general setting
are not only \Wtwo{}-hard, but \Wtwo{}-complete. 

Second, we have only shown that \mcccav{} is inapproximable
and \mcdcdv{} is trivially polynomial-time solvable.
For the other two problem variants, we do not know whether they 
can be approximated efficiently or not.

Another open question is whether there are \fpt{}-results for any natural combined parameters. 
As a starting point, we conjecture that all problem variants can be formulated
as a monadic second-order logic formula with length of at most $ f(k, b, m) $
(where $ k $ is the solution size, $ b $ is the maximum bundle size, $ m $ is the 
number of candidates, and $f$ is a computable function).
\Citet{courcelle2012graph} showed that every graph problem 
expressible as a monadic second-order logic formula $ \rho $ can be solved in 
$ g(|\rho|, \omega) \cdot |I| $ time, where $ \omega $ is the treewidth of the input graph 
and $|I|$ is the input size. 
Our conjecture would provide us with a fixed-parameter tractability result with
respect to the solution size, the maximum bundle size, the number of candidates,  
and the treewidth of our bundling graph~$G_\kappa$.

We have studied the Plurality rule exclusively.
Thus it is still open which of our results also hold for other
voting rules, especially for the Condorcet rule. 
Since with two candidates, the Condorcet rule is equivalent to the strict majority rule, 
we can easily adapt some of our results to work for the Condorcet rule as well. 
Other results (i.e., the Turing reductions) 
cannot be easily adapted to work for the Condorcet rule. 

\bibliographystyle{abbrvnat}

\newpage
\appendix

\begin{center}
  \Large \textbf{Appendix}
\end{center}

\section{Similarities in Complexity between the Problem Variants}
\label{app:similarities}
In this section we provide the theorems and proofs for results in which 
the four problem variants behave similarly in complexity which are
summarized in the main text as \cref{thm:similar_summary}. 

First, we provide hardness results with different constraints on the parameters of
the problem variants. 
Then, we show that the problem variants are fixed-parameter tractable with 
respect to the number of candidates. 

For this appendix, we introduce the Condorcet voting rule. 

A candidate $ c $ is a \emph{Condorcet winner} if it wins against 
every other candidate in a head-to-head contest \cite{condorcet1785essai}. Formally, $ c $ is a Condorcet winner if 
$ \forall c' \in C \setminus \{ c \}: \; 
|\Set{ v \in V}{ c \succ_v c' }| > |\Set{ v \in V}{ c' \succ_v c }| $.
Condorcet's voting rule returns a set consisting of the unique Condorcet winner if it exists.
Otherwise, it returns the empty set. 

Note that, for the Condorcet rule, the problem definitions stated in
\cref{sec:central_problem} need to be modified as the preferred winner
(loser) $ p $ needs to win (lose) the election evaluated by the Condorcet rule. 

\subsection{\Wtwo{}-Hard for the Solution Size}
\label{app:w2}

\citet{bulteau2015combinatorial} originally stated that \cccav{} is
\textsf{W[2]}-hard with respect to the solution size for the Plurality
and for the Condorcet voting rule. 
\cref{thm:w2} uses their proof concept and extends their result for 
the other three variants of the combinatorial voter control problem. 

\begin{theorem} \label{thm:w2}
	For both Plurality and Condorcet, \cccdv{}, \cdcav{} and \cdcdv{} are
	all \emph{\textsf{W[2]}}-hard with respect to the solution size $ k $, even if
	there are only two candidates and even if the bundling function $ \kappa $ is
	symmetric.
\end{theorem}

\begin{proof}
	We first consider the Plurality rule and provide a parameterized reduction
	from the \textsf{W[2]}-complete problem \domset{} parameterized by the 
	solution size $ h $.
	
        \probDef{\domset}
        {An undirected graph $G = (V(G),E(G))$ and a natural number $h \in \N$.}
        {Is there a \emph{dominating set} of size at most $ h $, that is, 
          a vertex subset $ U \subseteq V(G) $ with $ |U| \le h $ such that
          each vertex from $ V(G) \setminus U $ is adjacent to at least one
          vertex from $ U $?}
	
	Let $(G,h)$ be a \domset{} instance. We construct an election
	$E=(C,V)$ with $ C = \{ p, g \} $, where $p$ is our preferred candidate.
	We define the voter set for our three problem variants differently.
	
	For \cdcav{}:
	\begin{compactitem}[\qquad -]
	  	\item The registered voter set $V$ consists of $|V(G)|-1$ $p$-voters
	(and no $g$-voters).
		\item The unregistered voter set $W$ consists of one $g$-voter
	$w_i$ for each vertex $u_i \in V(G)$.
	\end{compactitem}
	
	For \cccdv{}, we define the voter set $ V $ such that
	\begin{compactitem}[\qquad -]
		\item $V$ consists of one $g$-voter $w_i$ for each vertex $u_i \in V(G)$ and
		\item no $p$-voters.\footnote{For simplicity, we assume that every candidate
		is a winner in an election without voters. However, if one requires a non-empty
		voter set, then one can easily adjust the construction by adding one $p$-voter $w_p$
		and one $g$-voter $w_g$ with $\kappa(w_p)=\{w_p\}$ and $\kappa(w_d)=\{w_d\}$.
		This makes the argumentation in the proof slightly more complicated.}
	\end{compactitem}
	
	For \cdcdv{}, we define the voter set $V$ such that
	\begin{compactitem}[\qquad -]
		\item $V$ consist of one $p$-voter $w_i$ for each vertex $u_i \in V(G)$ and
		\item one g-voter $w'$ with bundle $\kappa(w')=\{w'\}$.
	\end{compactitem}
	
	We define the bundle $\kappa(w_i)$ as the closed neighborhood of $u_i$,
	formally $\kappa(w_i) = {w_i} \cup \{w_j \mid \{u_i,u_j\} \in E(G)\}$. 
	Finalizing our construction, we set $ k \coloneqq h $.
	
	It is clear that our construction is both a polynomial reduction and a
	parameterized reduction with respect to $k$. Also, it is obvious that the
	bundling function used in the construction is symmetric.

	We exemplary show for \cccdv{} that there is a dominating set of size $h$ if
	and only if there is a subset $V'$ of size at most $k$. The other two 
	variants can be proven analogously.
	
	For the ``only if'' part, given a dominating set $U$ of size at most $h$, we
	define $V'$ to be the corresponding voter set, that is, $V'\coloneqq \{w_i \mid u_i \in U
	\}$. 
	It is clear that $|V'| \leq h = k$ and $p$ as well as $g$ becomes winner,
	because $g$ loses $|V(G)|$ points and has the same score as $p$.
	
	For the ``if'' part, given a subset of the voters $V' \subseteq V$ of size at
	most $k$ such that p is a winner of $E' = (C, V \setminus \kappa(V'))$, we
	define $U$ to be the set of vertices corresponding to the voters from $V'$, that is, $U\coloneqq \{ u_i \mid w_i \in V' \}$.
	It follows that $|U| \leq k = h$ and for each vertex $u_i \in V(G) \setminus U$ 
	there must be a vertex $u_j \in U$ which is a neighbor of $u_i$, 
	since otherwise we will still have some $g$-voters and $p$ will not become a
	winner.
	
	Finally, we need to consider the Condorcet rule. For two candidates, the 
	Condorcet rule is equivalent to the strict majority rule and, hence, the
	proof is analogous to the proof for the Plurality rule. 
	We only need to adapt the \cdcav{}
	such that the candidate $ p $ is the only winner. We can accomplish that by
	adding to the $ V $ one $p$-voter $w'$ with bundle $\kappa(w') = \{w'\}$.
\qed\end{proof}

\subsection{\Wone{}-Hard for the Solution Size}
\label{app:w1}
From \citet{bulteau2015combinatorial} we know that \cccav{} is
\textsf{W[1]}-hard with respect to the solution size for the
Plurality voting rule, even when the maximum bundle size is three.
We show this for three other variants of the combinatorial
voter control and additional for Condorcet's voting rule.

\begin{theorem} \label{thm:w1}
	For both Plurality and Condorcet, \cccav{}, \cccdv{}, \cdcav{} and \cdcdv{} are
	all \emph{\textsf{W[1]}}-hard with respect to the solution size $ k $, even
	if the maximum bundle size b is three and there are only two candidates.
\end{theorem}
\begin{proof}
	We first consider the Plurality rule and provide a parameterized reduction
	from the \textsf{W[1]}-complete problem \clique{} parameterized by the 
	solution size $ h $.
	
        \probDef{\clique}
        {An undirected graph $G = (V(G),E(G))$ and a natural number $h \in \N$.}
        {Is there a \emph{clique} of size at least $ h $, that is, a complete subgraph with $ h$ vertices?}

	Let $(G,h)$ be a \clique{} instance. Without loss of generality, we assume $h
	> 3$. (If not, it can be solved in polynomial time.) We construct an
	election $E=(C,V)$ with $ C = \{ p, g \} $, where $p$ is our preferred
	candidate.
	For each vertex $u \in V(G)$ we define one \emph{vertex voter} $w_u$ with the bundle
	$\kappa(w_u) = \{ w_u \}$.
	For each edge $e=\{ u, u' \} \in E(G)$ we define one \emph{edge voter} $w_e$
	with the bundle $\kappa(w_e) = \{ w_e, w_u, w_{u'} \}$.
	Now, we finalize the voter set definition which slightly differs for our three problem variants.
	
	For \cccdv{}, we define the voter set $V$ such that:
	\begin{compactitem}[\qquad -]
		\item Vertex voters are $p$-voters.
		\item Edge voters are $g$-voters.
		\item We add a set $D_p$~of dummy $p$-voters and a set $D_g$~of dummy $g$-voters
			with $\kappa(w) = \{ w \}$ for each $w \in D_p \cup D_g$.
			We set set the cardinalities of~$D_p$ and~$D_g$
			such that $ s_g(V) - s_p(V) = \binom{h}{2} - h$.
	\end{compactitem}
	
	For \cdcdv{}, we define the voter set $V$ such that
	\begin{compactitem}[\qquad -]
		\item Vertex voters are $g$-voters.
		\item Edge voters are $p$-voters.
		\item We add a set $D_p$~of dummy $p$-voters and a set $D_g$~of dummy $g$-voters
			with $\kappa(w) = \{ w \}$ for each $w \in D_p \cup D_g$.
			We set set the cardinalities of~$D_p$ and~$D_g$
			such that $ s_p(V) - s_g(V) = \binom{h}{2} - h - 1$.
	\end{compactitem}
	
	For \cdcav{}:
	\begin{compactitem}[\qquad -]
		\item Vertex voters are unregistered $g$-voters.
		\item Edge voters are unregistered $p$-voters.
		\item There are no other voters in the unregistered voter set $W$.
		\item The registered voter set $V$ consist of 
			$\binom{h}{2} - h - 1$ $p$-voters.
	\end{compactitem}
	
	Finally, we set $k\coloneqq\binom{h}{2}$.
	
	It is clear that our construction is both a polynomial reduction and a
	parameterized reduction with respect to $k$.
	
	We exemplary show for \cccdv{} that there is a clique $U$	of size at least
	$h$ if and only if there is a subset $V'$ of size at most $k$ such that p
	becomes a winner of $E$. The other two variants can be proven analogously.
	
	For the ``only if'' part, given a clique~$U$ of size~$h$, we construct 
	$V'$ by adding to it any edge voter $w_e$ with $e \in E(G[U])$.
	It is clear that $|V'| \leq \binom{h}{2} = k$.
        Observe that candidate~$p$ as well as candidate~$g$ become winners.
	Candidate~$g$ loses $\binom{h}{2}$~points and candidate~$p$ loses $h$~points.
        Thus, $g$ and $p$ have the same score.
		
	For the ``if'' part, given a subset of the voters $V' \subseteq V$ of size at
	most $k$ such that p is a winner of $E' = (C, V \setminus \kappa(V'))$, we
	define $U$ to be the set of vertices corresponding to the voters from $V'$,
	that is, $U\coloneqq \{ u \in V(G) \mid w_e \in V'$ and $ u \in e \}$.
	We observe that deletion of vertex voters doesn't reduce the score of $g$
	and removing a vertex voter from~$V'$ would lead to a smaller solution.
	Hence, we can assume that~$V'$ does not contain any vertex voters.
	In order to reduce the score of $g$, enough edge voters must be removed,
	but a certain amount of vertex voters will be removed as well since they
	are in the bundles of the edge voters.
	We denote the number of indirectly removed vertex voters be~$x$.
	Clearly $x \le h$, because otherwise $p$ loses more than $h$~points,
	$g$~loses at most $\binom{h}{2}$~points, and
	$g$~remains the only winner.
	Assume towards a contradiction that $x \le h-2$.
	The score of~$g$ decreases by at least $\binom{h}{2} - h$ in $E'$ (compared to~$E$)
	so that $V'$ contains at least $\binom{h}{2} - h$ edge voters.
	However, $\binom{x}{2}\le \binom{h-2}{2} < \binom{h}{2} - h$ for any $h>3$.
	Hence, $x\ge h-2$ implying that there are at least $\binom{h}{2} - 2$ edge voters.
	Now, assume towards a contradiction that $x = h-1$.
	Then, $\binom{x}{2}\le \binom{h-1}{2} < \binom{h}{2} - 2$ for any $h>3$.
	Thus, $x=h$ implying that there are exactly $k=\binom{h}{2}$ edge voters in $V'$
	with altogether $h$~different vertex voters in their bundles.
	In this case, $U$ is a clique of size $h$, since otherwise we cannot have $\binom{h}{2}$ edges
	incident to $h$~vertices.

	Finally, we need to consider the Condorcet rule. For two candidates, the 
	Condorcet rule is equivalent to the strict majority rule and, hence, the
	proof is analogous to the proof for the Plurality rule. 
	We only need to adapt the \cdcav{}
	such that candidate $ p $ is the only winner. We can accomplish that by
	adding to the $ V $ one $p$-voter $d$ with bundle $\kappa(d) = \{d\}$. 	
	\ccccav{} can be proved analogously.
\qed\end{proof}

\subsection{Fixed-Parameter Tractability for the Number of Candidates}
\label{app:ilp}
\citet{bulteau2015combinatorial} provide an integer linear program (ILP)
that solves Plurality- and Condorcet-\cccav{} for the case when the bundling function is
anonymous (see \cref{sec:preliminaries} for the corresponding definition) and exploit Lenstra's theorem to show fixed-parameter tractability with respect to the number of candidates.
Their idea is to utilize the fact that for anonymous bundling functions, 
voters with the same preference order ``lead'' the same bundle of voters and ``follow'' the same voter as well.
Thus, with $m$~candidates, we will have at most $m!$ different bundles.
By this observation, they introduce $\ON(m!)$ variables, one variable for a bundle, 
to encode whether to select a bundle to the solution.  
Indeed, as long as the bundling function is anonymous, 
the same idea applies to the remaining three combinatorial voter control variants.
Although the technique is analogous, 
we provide the corresponding ILPs for \cref{thm:ilp,thm:ilp2,thm:ilp3} and show the correctness for the sake of completeness.

\begin{theorem} \label{thm:ilp}
	For both Plurality and Condorcet, \cccdv{} is 
	fixed-parameter tractable with respect to the number $ m $
	of candidates, if the bundling function $ \kappa $ is anonymous.
\end{theorem}

\begin{proof}
	Given a \cccdv{} instance~$((C,V), \kappa, p, k)$ with $m$ candidates and anonymous bundling function~$\kappa$, 
        we construct an integer linear program (ILP) with at most $ \ON(m!)$ variables and at most $ \ON(|V|+m) $ constraints for \cccdv{} similar to
	\citet{bulteau2015combinatorial} for \cccav{}.
	Fixed-parameter tractability follows because every ILP with $ \rho $ variables and 
	$L$ input bits is solvable in $O(\rho^{2.5 \rho + o(\rho)} L)$ time \cite{lenstra1983integer}.

	Since $m$ candidates are given, there are at most $m!$ voters with pairwise
	different preference orders. We denote these as $ \succ_1, \succ_2,
	\dots, \succ_{m!} $ and note that there are at most $m!$ different
	bundles, because $\kappa$ is leader-anonymous.

	We will use the following notation for the construction of the ILP:
	\begin{compactenum}
	\item Define $ \kappa(\succ_i) $ as the set of preference orders of the voters
		included in the bundle of the voters with preference order $\succ_i$; note that by the anonymity, if $\kappa(\succ_i)$ contains a preference order~$\succ$, 
                then every voter with preference order~$\succ$ is in the bundle of the voter that has preference order~$\succ_i$. 
	\item Define $\kappa^{-1}(\succ_j) \coloneqq \{\succ_i \mid \succ_j \in \kappa(\succ_i)\} $ as the set of preference orders that include $\succ_j$ in their bundles.
	\item Define $ N_i $ as the number of voters with preference order $ \succ_i $ in $ V $.
	\item For each candidate~$a\in C$, define $ F(a) $ as the set of preference orders in which $ a \in C $ is ranked first and
	\item let $ s(a) $ be the initial score in election~$(C, V)$.
	
	\end{compactenum}
        To encode a solution~$W'$, 
        for each preference order $ \succ_i, i \in [m!] $, we introduce two Boolean
	variables, $x_i$ and $y_i$.
	The intended meaning of $x_i = 1$ is that the
	sought solution contains a voter with preference order $ \succ_i $.
	The intended meaning of $ y_i = 1 $ is that $ \kappa(V') $
	contains a voter with preference order $ \succ_i $.

	Now, we are ready to state the integer linear problem.
	Note that it suffices to find a feasible solution. Thus, we 
	do not need to specify any objective function.
  \allowdisplaybreaks
\begin{flalign}
    \label{ilp:solutionbound}
    \sum_{i\in [m!]}{x_{i}} &\le k, & \\
    \label{ilp:voterbound}
    x_{i} &\le N_i, & \forall i \in [m!]\quad\quad \\
    \label{ilp:y-one}
    \sum_{\mathord{\succ_i} \in \kappa^{-1}(\mathord{\succ_j})}x_i &\leq m! \cdot y_j, & \forall j \in [m!]\quad\quad \\
    \label{ilp:y-two}
    \sum_{\mathord{\succ_i} \in \kappa^{-1}(\mathord{\succ_j})}x_i &\geq y_j, & \forall j \in [m!]\quad\quad \\
    \label{ilp:winning}
    s(p) - \sum_{\succ_j \in F(p)} N_j\cdot y_j &\geq s(a) - \sum_{\succ_j \in F(a)} N_j\cdot y_j,& \forall a \in C\setminus \{p\}\quad\quad\\
    \label{ilp:binary_vars}
    x_{i}, y_{i} & \in \{0,1\}, & \forall i \in [m!]\quad\quad 
  \end{flalign}

	Constraint~\eqref{ilp:solutionbound} ensures that at most $k$ voters are added to the
	solution.
	Constraint~\eqref{ilp:voterbound} ensures that the voters added to the
	solution are indeed present in $ V $.
	Constraints \eqref{ilp:y-one} and \eqref{ilp:y-two} ensure that variables 
	$ y_j, 1 \le j \le m! $, have correct values. Indeed, if for some preference
	order $ \succ_i $ we have $ x_i=1 $ and $ \succ_j \in \kappa(\succ_i) $, then
	constraint~\eqref{ilp:y-one} ensures that $ y_j=1 $.
	On the other hand, if for some
	preference order $\succ_j$ we have that for each preference order $\succ_i$
	with $ \succ_j \in \kappa(\succ_i) $ it holds that $ x_i=0 $, then 
	constraint~\eqref{ilp:y-two} ensures that $y_j=0$.
	Constraint~\eqref{ilp:winning} ensures that $ p $ has a (plurality) score which
	is at least as high as the score of every candidate (which makes $ p $ a
	winner).
	Clearly, there is a solution for this integer linear program if and only if
	there is a solution for \cccdv{} with the Plurality rule.

	For the case of the Condorcet rule, we need to define the following additional
	parameters:
	Let $ s(a, b) \coloneqq | \{ v \in V \mid a \succ_v b \} | $ denote the number of
	voters that prefer candidate $ a $ over candidate $	b $ and $ P(a, b) $ denote
	the set of preference orders in which $ a $ is preferred to $ b $.
	We modify only constraint \eqref{ilp:winning} as follows:	
	\begin{align}
		s(p,a) - \sum_{\mathclap{\succ_j \in P(p,a)}} N_j \cdot y_j & \ge s(a,p) - \sum_{\mathclap{\succ_j
		\in P(a,p)}} N_j \cdot y_j +1
	& \forall a \in C \setminus \{p\}	
	\label{ineq:ilp13}
	\end{align}
	This ensures that $p$ can beat every other given candidate $ a $ in a 
	head-to-head contest if and only if there is a solution to the ILP.

        As for the running time, it is clear that both ILPs have $O(m\!)$ variables and $O(|V|+m\!)$ constraints.
        By the famous result of \citet{lenstra1983integer}, fixed-parameter tractability follows.
\qed\end{proof}

\begin{theorem}
	\label{thm:ilp2}
	For both Plurality and Condorcet, \cdcav{} is fixed-parameter tractable
	with respect to the number $ m $ of candidates, if the bundling function
	$ \kappa $ is anonymous.
\end{theorem}
\begin{proof}
	As in the proof for \cref{thm:ilp}, we construct an integer linear program (ILP)
	with at most $ \ON(m!) $ variables and constraints, and use the 
	same notation, except the following two:
        \begin{itemize}
          \item $ N_i $ denotes the number of voters with preference order $ \succ_i $ in
          $ W $.
          \item The intended meaning of $ y_i = 1 $ is that $ \kappa(W') $
          contains a voter with preference order $ \succ_i $.
        \end{itemize}
        
        The constraints for the ILP are as follows: 
        \allowdisplaybreaks
	\begin{flalign}
	\sum_{i \in [m!]} x_i & \le k, & \\	
	x_i & \le N_i, & \forall i \in [m!] \\
	\sum_{\succ_i \in \kappa^{-1}(\succ_j)} x_i & \le m! \cdot y_i, &  \forall j \in [m!]\\
	\sum_{\succ_i \in \kappa^{-1}(\succ_j)} x_i & \ge y_j, & \forall j \in [m!]\\
        \alpha_a \Big( s(p) + \sum_{\mathclap{\succ_j \in F(p)}} N_j \cdot y_j +1 \Big)
		& \le \alpha_a \Big( s(a) + \sum_{\mathclap{\succ_j \in F(a)}} N_j \cdot y_j \Big),
		& \forall a \in C \setminus \{p\} \label{ineq:ilp8}	\\
		\sum_{a \in C \setminus \{p\} } \alpha_a & \ge 1, \label{ineq:ilp9}\\
	x_i, y_i, \alpha_a & \in \{0,1\},  & \forall i \in [m!]	\label{ineq:ilp51}
      \end{flalign}
	If $ \alpha_a = 0 $, constraint~\eqref{ineq:ilp8} is valid. If $ \alpha_a = 1
	$, constraint~\eqref{ineq:ilp8} ensures that $ p $ loses against candidate $ a
	$.
	Constraint~\eqref{ineq:ilp9} ensures that at least one of the Boolean variables 
	$ \alpha_a $ has value $ 1 $ and, therefore, there exists at least one
	candidate $ a \in C $ such that $ a $ wins against $ p $.

	For the Condorcet rule, alter constraint~\eqref{ineq:ilp8} as
	follows:
		
	\begin{align}
		\alpha_a \Big( s(p,a) + \sum_{\mathclap{\succ_j \in O(p,a)}} N_j \cdot y_j +1 \Big)
		& \le \alpha_a \Big( s(a,p) + \sum_{\mathclap{\succ_j \in O(a,p)}} N_j \cdot y_j \Big),
		& \forall a \in C \setminus \{p\}
		\label{ineq:ilp12}
	\end{align}
	
	Here $ s(a, b) = | \{ w \in W \mid a \succ_w b \} | $ is the number of unregistered voters that prefer
	candidate $ a $ over candidate $ b $.
	
        Constraint~\eqref{ineq:ilp9} ensures that at least one of the Boolean
        variables $ \alpha_a $ has value $ 1 $ while constraint~\eqref{ineq:ilp12}
        ensures that $ p $ loses against at least one of the candidates $ a $ in a head-to-head contest.

        We omit the reasoning for the running time as it is the same as the one shown for \cref{thm:ilp}.
\qed\end{proof}
\begin{theorem}
	\label{thm:ilp3}
	For both Plurality and Condorcet, \cdcdv{} is fixed-parameter-tractable
	with respect to the number $ m $ of candidates, if the bundling function
	$ \kappa $ is anonymous. 
\end{theorem}
\begin{proof}
	This ILP is almost the same as the one for \cdcav{}. The only
	difference is that we use ``$-$'' instead of ``+'' in
	constraint~\eqref{ineq:ilp8}:
	
	\begin{align}
	\alpha_a \Big( s(p) - \sum_{\mathclap{\succ_j \in F(p)}} N_j \cdot y_j +1 \Big)
		\le \alpha_a \Big( s(a) - \sum_{\mathclap{\succ_j \in F(a)}} N_j \cdot y_j \Big),
		 \forall a \in C \setminus \{p\}	
		\label{ineq:ilp10}
	\end{align}
	
	Note that, in this problem variant, we do not try to find a subset of unregistered 
	voters $ W' \subseteq W $ to add to the election, but a subset of the 
	registered voters $ V' \subseteq V $ to remove from the election. 
	The definitions for $ N_i $ and $ y_i $ change accordingly.

	For the Condorcet rule, alter 
	constraint~\eqref{ineq:ilp10} as follows:
	\begin{align}
		\alpha_a \Big( s(p,a) - \sum_{\mathclap{\succ_j \in O(p,a)}} N_j \cdot y_j +1 \Big)
		 \le \alpha_a \Big( s(a,p) - \sum_{\mathclap{\succ_j \in O(a,p)}} N_j \cdot y_j \Big), \forall a \in C \setminus \{p\}
	\end{align}
	Here $ s(a, b) = | \{ w \in W \mid a
	\succ_w b \} | $ is the number of unregistered voters that prefer
	candidate $ a $ over candidate~$b$.
        We omit the reasoning for the running time as it is the same as the one shown for \cref{thm:ilp}.
\qed\end{proof}

\section{Missing Proofs}
\appendixProofText

\label{app:inapproximability}

\label{app:disjoint}

\end{document}